\documentclass[12pt,reqno]{amsart}

\newcommand\version{September 23, 2018}

\usepackage{amsmath,amsfonts,amsthm,amssymb,amsxtra}
\usepackage{dsfont}

\newtheorem{theorem}{Theorem}
\newtheorem{proposition}[theorem]{Proposition}
\newtheorem{lemma}[theorem]{Lemma}
\newtheorem{corollary}[theorem]{Corollary}

\theoremstyle{definition}

\newtheorem{assumption}[theorem]{Assumption}

\theoremstyle{remark}

\newtheorem{remark}[theorem]{Remark}

\newcommand{\ch}{\mathord{\mathfrak h}}

\renewcommand{\epsilon}{\varepsilon}

\newcommand{\ii}{{\rm i}}

\renewcommand{\phi}{\varphi}
\newcommand{\R}{\mathbb{R}}

\newcommand{\Z}{\mathbb{Z}}
\newcommand\1{{\ensuremath {\mathds 1} }}

\DeclareMathOperator{\ran}{ran}
\DeclareMathOperator{\re}{Re}
\DeclareMathOperator{\spec}{spec}

\begin{document}

\title[Critical temperature --- \version]{The BCS critical temperature in a weak external electric field via a linear two-body operator}

\author[R. L. Frank]{Rupert L. Frank}
\address[R. L. Frank]{Mathematisches Institut der Universit\"at M\"unchen, Theresienstr. 39, 80333 M\"unchen, Germany, and Mathematics 253-37, Caltech, Pasadena, CA 91125, USA}
\email{rlfrank@caltech.edu}

\author[C. Hainzl]{Christian Hainzl}
\address[C. Hainzl]{Mathematisches Institut, Universit\"at T\"ubingen, Auf der Morgenstelle 10, 72076 T\"ubingen, Germany}
\email{christian.hainzl@uni-tuebingen.de}

\begin{abstract}
We study the critical temperature of a superconductive material in a weak external electric potential via a linear approximation of the BCS functional. We reproduce a similar result as in \cite{FHSS2} using the strategy introduced in \cite{FHaiLa}, where we considered the case of an external constant magnetic field. 
\end{abstract}

\dedicatory{Dedicated to Herbert Spohn on the occasion of his seventieth birthday}

\maketitle

\renewcommand{\thefootnote}{${}$} \footnotetext{\copyright\, 2018 by
  the authors. This paper may be reproduced, in its entirety, for
  non-commercial purposes.}

\section{Introduction and main result}

\subsection{Objective and background}

In this paper we want to consider a linear two-body operator which determines the critical temperature of a superconductive or superfluid system. 
This linear operator was studied recently in connection with the influence of a constant magnetic field on the critical temperature \cite{FHaiLa}. 
The analysis of this operator was significantly complicated by the unboundedness of the magnetic vector potential as well as the non-commutativity of 
the components of the magnetic momentum. For this reason we want to present here the method of \cite{FHaiLa} in the simplified situation where the external field consists of an electric potential. 

We have the following situation in mind. Two particles interact via a two body potential $-2V(x-y)$ and both particles are placed in an external electric potential $h^2 W(hx)$, where $h>0$ is a small parameter. Thus, the external field is weak of order $h^2$ and varies on the scale of order $1/h$, whereas both the strength and the scale of the interaction are of order one determined by $V$. The energy is given by the linearized BCS (Bardeen--Cooper--Schrieffer) functional at positive temperature $T = 1/\beta$. 

Therefore we are interested in the infimum of the spectrum of the two-body operator
\begin{equation}
\label{eq:twobodyk}
\frac{p_x^2 +h^2 W(hx)  +p_y^2 +h^2 W(hy) - 2\mu}{\tanh \left( \frac\beta 2 \left( p_x^2 +h^2 W(hx) - \mu \right)\right) + \tanh \left( \frac\beta 2 \left(p_y^2 +h^2 W(hy) - \mu \right) \right)} - V(x-y)
\end{equation}
acting in
$$
L_{\rm symm}^2(\R^3\times\R^3) = \left\{\alpha\in L^2(\R^3\times\R^3):\ \alpha(x,y)=\alpha(y,x) \ \text{for all}\ x,y\in\R^3 \right\}.
$$
Here $p_x=-\ii\nabla_x$ and $p_y=-\ii\nabla_y$. The interaction potential $-2V(x-y)$ between the two particles is assumed to be spherically symmetric, i.e., to depend only on the distance $|x-y|$. (We will also assume that the interaction potential is non-positive and the minus sign, as opposed to the more usual plus sign, will simplify some formulas.) Moreover, $\mu\in\R$ is the chemical potential. We are interested in the dependence of the operator on two parameters, namely, the inverse temperature $\beta>0$ and the scale ratio $h>0$. More precisely, we are interested in identifying regimes of temperatures $T=\beta^{-1}$ such that the infimum of the spectrum of the above operator is positive or negative for all sufficiently small $h>0$. 

As we explained in detail in \cite{FHaiLa} and will repeat below, the motivation for this question comes from the BCS theory of superconductivity and the operator \eqref{eq:twobodyk} arises through the linearization of the Bogolubov--de Gennes equation around the normal state. Therefore, the question whether the infimum of the spectrum of the operator \eqref{eq:twobodyk} is positive or negative corresponds to the local stability or instability of the normal state. In that sense it is not hard to imagine that the BCS critical temperature corresponds to the value of $T$ for which the infimum of the spectrum of this operator is exactly zero. 

To describe our main result we introduce the effective one-body operator
\begin{equation}
\label{eq:onebody}
\frac{(-\ii\nabla_r)^2 - \mu}{\tanh \left( \frac\beta 2 \left( (-\ii\nabla_r)^2 - \mu \right)\right)} - V(r)
\end{equation}
acting in
$$
L^2_{\rm symm}(\R^3) = \{ \alpha \in L^2(\R^3):\ \alpha(-r)=\alpha(r) \ \text{for all}\ r\in\R^3\} \,.
$$
Later on, we will see that the variable $r\in\R^3$ arises as the relative coordinate $r=x-y$ of the two particles at $x$ and $y$. We will \emph{assume} that the operator $|(-\ii\nabla_r)^2-\mu| -V(r)$ has a negative eigenvalue. Then it is easy to see (see, e.g., \cite{HHSS}) that there is a unique $\beta_c\in (0,+\infty)$ such that the operator \eqref{eq:onebody} is non-negative for $\beta\leq\beta_c$ and has a negative eigenvalue for $\beta>\beta_c$. Let $T_c = \beta_c^{-1}$. Then our main result is, roughly speaking, that the infimum of the spectrum of the two-particle operator \eqref{eq:twobodyk} is negative for $T \leq T_c+c_0 h^2 + o(h^2)$ and positive for $T\geq T_c +c_0 h^2 - o(h^2)$. Here $c_0$ is a positive constant which we compute explicitly in terms of the zero-energy ground state of \eqref{eq:onebody} at $\beta=\beta_c$. (In fact, $c_0 = -T_c D_c$ with $D_c$ from \eqref{eq:dc}.) Thus, the external electric field $h^2 W(hx)$ changes the critical temperature by an amount $c_0 h^2 + o(h^2)$. Informally (that is, ignoring issues like the possible non-uniqueness of a critical temperature), this says that
$$
T_c(h)= T_c + c_0 h^2 + o(h^2) \,.
$$

The mathematical challenge of this problem is that low energy states of the two-particle operator \eqref{eq:twobodyk} exhibit a two-scale structure. As function of the relative coordinate $r=x-y$ and the center of mass coordinate $X=(x+y)/2$ they vary on a scale of order one with respect to $r$ and on a (much larger) scale of order $1/h$ with respect to $X$. The variation on the former scale is responsible for the leading order term $T_c$ for the critical temperature, whereas the variation on the latter scale is responsible for the subleading correction $c_0 h^2$. This subleading correction is determined by an effective linear Ginzburg--Landau functional which emerges on the macroscopic scale $1/h$ determined by the external potential. We hereby recover a similar result for the critical temperature as in the full non-linear BCS theory in \cite{FHSS2}. This is of course not unexpected since we deal with the second derivative around the normal state of the BCS functional. 

The work \cite{FHSS2} relied on \cite{FHSS} where the Ginzburg--Landau functional was derived from the BCS functional close to the critical temperature by means of a rather intricate proof. In view of this, the goal of the present paper is twofold. First, we explain the strategy from \cite{FHaiLa} in a simpler setting, and second, we derive the linearized Ginzburg--Landau equation in a simpler way as in the full non-linear case \cite{FHSS}. One difference compared to the work \cite{FHSS,FHSS2} is the fact that we do not restrict ourselves to a finite box and therefore omit the periodicity assumptions. Further, we work in relative and center-of-mass coordinates which is natural in terms of the before mentioned two-scale structure.

As in \cite{FHaiLa} we will not work directly with the two-particle operator \eqref{eq:twobodyk}, but rather with its Birman--Schwinger version.

Before we describe the precise set-up of our analysis, we would like to stress that in this paper we work with the BCS functional and its linearization around the normal state. This should not be confused with what is often called the BCS Hamiltonian or the BCS model and which was investigated, for instance, by Haag, Thirring and Wehrl from the point of view of algebraic quantum field theory. The BCS Hamiltonian is a many-body Hamiltonian which corresponds to a regularization of a $\delta$ interaction. The BCS functional arises as an effective non-linear functional by restricting the BCS Hamiltonian to quasi-free states and dropping the direct and exchange terms. We do allow, however, for more general interaction potentials. It remains an open problem to understand from a mathematically rigorous point of view the relation between the BCS functional and many-body quantum mechanics.  Nevertheless, our analysis leads to quantitative estimates which agree with physics.


\subsection{Model and main result}

Our model has the following ingredients.

\begin{assumption}\label{ass2}
(1) External electric potential $h^2 W(hx)$ such that $W \in W^{1,\infty}(\R^3)$.\\
(2) Inverse temperature $\beta = T^{-1}>0$\\
(3) Chemical potential $\mu\in\R$\\
(4) Non-negative, spherically symmetric interaction potential $V$ such that $V\in L^\infty(\R^3)$ and $|r|V\in L^\infty(\R^3)$
\end{assumption}

We recall that the Sobolev space $W^{1,\infty}(\R^3)$ consists of all bounded, Lipschitz continuous functions with a finite global Lipschitz constant.

The non-negativity assumption on $V$ is for technical convenience.
To simplify notation and since the precise meaning is always clear from the context, we use the same symbol $V$ also for the corresponding multiplication operators on $L^2_{\rm symm}(\R^3)$ (i.e., $(V\alpha)(r) = V (r)\alpha(r)$) and on $L^2_{\rm symm}(\R^3\times\R^3)$ (i.e., $(V\alpha)(x, y) = V (x-y)\alpha(x, y)$).

The corresponding single-particle Hamiltonian, acting in $L^2(\R^3)$, is defined by
\begin{equation}
\label{eq:hw}
\ch_W = p^2 +h^2 W(hx) - \mu \,.
\end{equation}
with the notation $p=-\ii\nabla$. The locations of the two particles are represented by coordinates $x,y\in\R^3$. If we want to emphasize the variables on which the operators act, we write
$$
\ch_{W,x} = p_x^2 + h^2 W(hx) -\mu \,,
\qquad
\ch_{W,y} = p_y^2 + h^2 W(hy) -\mu \,.
$$
 As in \cite{FHaiLa} we introduce a function $\Xi_\beta:\R^2\to\R$ by
$$
\Xi_\beta(E,E') := \frac{\tanh\frac{\beta E}{2}+\tanh\frac{\beta E'}{2}}{E+E'}
$$
if $E+E'\neq 0$ and $\Xi_\beta(E,-E) = (\beta/2)/\cosh^2(\beta E/2)$. Since the operators $\ch_{W,x}$ and $\ch_{W,y}$ commute, we can define the operator
$$
L_{T,W} = \Xi_\beta(\ch_{W,x},\ch_{W,y}) \,.
$$
We will always consider this operator in the Hilbert space $L^2_{\rm symm}(\R^3\times\R^3)$. Note that, with this notation, the operator in \eqref{eq:twobodyk} can be written as $L_{T,W}^{-1}-V$.

Next, in order to formulate our assumption on the critical temperature, we introduce the function $\chi_\beta:\R\to\R$ by
$$
\chi_\beta(E) := \frac{\tanh\frac{\beta E}2}{E}
$$
and set $\chi_\infty(E):=|E|^{-1}$. We consider the compact operator
$$
V^{1/2} \chi_\beta(p_r^2-\mu) V^{1/2}
$$
in $L^2_{\rm symm}(\R^3)$, where
$$
p_r = -\ii\nabla_r
$$
denotes the momentum operator. (The operator $\chi_\beta(p_r^2-\mu)$ is denoted by $K_T^{-1}$ in \cite{HHSS} and several works thereafter.) 

\begin{assumption}\label{ass0}
$\sup\spec V^{1/2}\chi_\infty(p_r^2-\mu)V^{1/2}>1$.
\end{assumption}

Since $\beta\mapsto\chi_\beta(E)$ is strictly increasing for each fixed $E\in\R$, Assumption \ref{ass0} implies that there is a unique $\beta_c\in (0,\infty)$ such that
\begin{align*}
\sup \spec V^{1/2} \chi_\beta(p_r^2-\mu) V^{1/2} & \leq 1 \qquad\text{if}\ \beta\leq\beta_c \,,\\
\sup \spec V^{1/2} \chi_\beta(p_r^2-\mu) V^{1/2} & > 1 \qquad\text{if}\ \beta>\beta_c \,.
\end{align*}
We set $T_c=\beta_c^{-1}$. Note that the operator $V^{1/2} \chi_{\beta_c}(p_r^2-\mu) V^{1/2}$ has eigenvalue $1$.

\begin{assumption}\label{ass1}
The eigenvalue $1$ of the operator $V^{1/2} \chi_{\beta_c}(p_r^2-\mu) V^{1/2}$ is simple.
\end{assumption}

We denote by $\phi_*$ a normalized eigenfunction of $V^{1/2} \chi_\beta(p_r^2-\mu) V^{1/2}$ corresponding to the eigenvalue $1$ which, by assumption, is unique up to a phase. Since $p_r^2$ and $V$ are real operators, so is $V^{1/2} \chi_\beta(p_r^2-\mu) V^{1/2}$ and we can assume that $\phi_*$ is real-valued.

The spherical symmetry of $V$ from Assumption \ref{ass2} and the non-degeneracy from Assumption \ref{ass1} imply that $\phi_*$ is spherically symmetric.

From a physics point of view, Assumption \ref{ass1} restricts us to potentials giving rise to s-wave superconductivity. It is known that this assumption is fulfilled for a large class of potentials, including those which have a non-negative Fourier transform \cite{HaSe}. For partial results in the case where Assumption \ref{ass1} is violated, we refer to \cite{FrLe}.

As the final preliminary before stating our main result, we will introduce some constants. They are defined in terms of the auxiliary functions
\begin{align}\label{eq:auxiliary}
g_0(z) & = \frac{\tanh(z/2)}{z} \,,\notag \\
g_1(z) & = \frac{e^{2z}-2ze^z-1}{z^2(e^z+1)^2} = \frac1{2z^2} \frac{\sinh z-z}{\cosh^2(z/2)}\,, \notag \\
g_2(z) & = \frac{2e^z (e^z-1)}{z(e^z+1)^3} = \frac1{2z} \frac{\tanh(z/2)}{\cosh^2(z/2)}\,,
\end{align}
as well as the function
\begin{equation}
\label{eq:t}
t(p) := \|\chi_{\beta_c}((-\ii\nabla_r)^2-\mu) V^{1/2}\phi_*\|^{-1}\  2 (2\pi)^{-3/2} \int_{\R^3} dx\, V(x)^{1/2} \phi_*(x) e^{-\ii p\cdot x} \,.
\end{equation}
(The prefactor in front of the integral is irrelevant for us and only introduced for consistency with the definition in \cite{FHSS2}.) We now set
\begin{align}\label{eq:glcoeff}
\Lambda_0 & := \frac{\beta_c^2}{16} \int_{\R^3} \frac{dp}{(2\pi)^3}\, |t(p)|^2 \left(g_1(\beta_c(p^2-\mu)) + \frac23 \beta_c p^2 g_2(\beta_c(p^2-\mu)) \right) \,,\\
\Lambda_1 & := \frac{\beta_c^2}{4} \int_{\R^3} \frac{dp}{(2\pi)^3}|t(p)|^2\, g_1(\beta_c(p^2-\mu)) \, \,, \\
\Lambda_2 & := \frac{\beta_c}8 \int_{\R^3} \frac{dp}{(2\pi)^3}\, |t(p)|^2 \cosh^{-2}(\beta_c(p^2-\mu)/2) \,.
\end{align}
The constants $\Lambda_0$ and $\Lambda_2$ are positive (for a proof for $\Lambda_0$ see \cite{FHSS}). Note that the quotient $\Lambda_0/\Lambda_2$, which will appear in our main result, has the dimension of an inverse temperature.

We set 
\begin{equation}
\label{eq:dc}
D_c := \frac{\Lambda_0}{\Lambda_2}\, \inf\spec \left( p_X^2 + \frac {\Lambda_1}{\Lambda_0} W(X) \right),
\end{equation}
where the operator on the right side is considered as an operator in $L^2(\R^3)$ and where $p_X=-\ii\nabla_X$.

The following is our main theorem.

\begin{theorem}\label{main}
Under assumptions \ref{ass2}, \ref{ass0} and \ref{ass1} the following holds.
\begin{enumerate}
\item[(1)] Let $0<T_1<T_c$. Then there are constants $h_0>0$ and $C>0$ such that for all $0< h \leq h_0$ and all $T_1\leq T< T_c (1 - h^2 D_c) -C h^3$ one has
$$
\inf_{\Phi} \langle\Phi, (1- V^{1/2} L_{T,W} V^{1/2})\Phi\rangle <0 \,.
$$
\item[(2)] There are constants $h_0>0$ and $C>0$ such that for all $0< h \leq h_0$ and all $T> T_c(1- h^2 D_c) +  C h^{5/2}$ one has
$$
\langle\Phi, (1- V^{1/2} L_{T,W} V^{1/2})\Phi\rangle >0 \,,
$$
unless $\Phi=0$.
\end{enumerate}
\end{theorem}

\begin{remark}
Let us restate this theorem in a heuristic form. Informally, we think of the critical temperature $T_c(h)$ as the value of the parameter $T$ such that
$$
\sup \, {\rm spec} V^{1/2} L_{T,W} V^{1/2} = 1 \,.
$$
This is not a precise definition because in contrast to the one-body operator $V^{1/2} \chi_\beta(p_r^2-\mu) V^{1/2} $ it is not clear whether the two-body operator $V^{1/2} L_{T,W} V^{1/2}$, or at least the infimum of its spectrum, is monotone in $T$ and therefore the uniqueness of the value of $T$ such that $\sup \, {\rm spec} V^{1/2} L_{T,W} V^{1/2} = 1$ is not guaranteed. Ignoring this issue, as well as some technicalities connected with $T_1$ in part (1) which we discuss below, we see that our main theorem says that
$$
T_c(h)= T_c(1 - D_c h^2) + o(h^2) \,.
$$
Note that concerning the potential non-uniqueness of the critical temperature the theorem implies that, if it occurs at all, it occurs only in a temperature interval of size $o(h^2)$.
\end{remark}

\begin{remark}
Observe that $D_c$ can have either sign, depending on $W$. Thus, an external electric field $h^2 W(hx)$ can both raise and lower the critical temperature by an amount of order $h^2$. This is in contrary to the influence of magnetic fields where the critical temperature always goes down.
\end{remark}

\begin{remark}
Let us compare our results here with those in \cite{FHSS2} where we also computed the shift of the critical temperature. The results of \cite{FHSS2} concern a definition of the critical temperature in the non-linear  BCS functional, whereas here we base our definition of critical temperature on a quadratic approximation to the BCS functional around the normal state. Both notions lead to the same result to order $h^2$. A minor difference is that the setting in \cite{FHSS2} is a finite sample whereas here we work on the whole space. Technically, the methods of proof in the two approaches are quite different.
\end{remark}

\begin{remark}
The assumption in part (1) that the temperature is bounded away from zero is probably only technical. Note, however, that our result is valid for arbitrarily small $T_1>0$, as long as it is uniform in $h$. The reason for this restriction is that our expansions diverge as the temperature goes to zero. Remarkably, there is no such restriction in part (2) of the theorem.
\end{remark}

\begin{remark}
Let us emphasize that our definition of the critical temperature $T_c$ coincides with that in \cite{HHSS} (and therefore with that in \cite{FHSS,FHSS2}) and that our Assumptions \ref{ass0} and \ref{ass1} coincides with \cite[Assumption 2]{FHSS}. This is a consequence of the Birman--Schwinger principle, which also implies that, if $\alpha_*$ denotes a normalized, real-valued eigenfunction of the operator \eqref{eq:onebody}, then
$$
V^{1/2}\alpha_* = \pm \|\chi_{\beta_c}((\ii\nabla_r)^2-\mu) V^{1/2}\phi_*\|^{-1} \phi_* \,.
$$
(To get the normalization constant, we apply $\chi_{\beta_c}((-\ii\nabla_r)^2-\mu) V^{1/2}$ to both sides and use the equation for $\alpha_*$ and its normalization.)
\end{remark}


\begin{remark}
In the physics literature the two-body interaction $V$ is usually replaced by a local contact interaction. With this modification the linear two-body operator \eqref{eq:twobodyk} was studied earlier in the literature in particular in the school by Gorkov and co-authors. In the presence of a constant magnetic fields this operator was used by Werthamer et al. \cite{HeWe, WeHeHo} in their study of the upper critical field. This approach was later extended in different directions, see e.g., \cite{SchSch, La, La2}. In particular, \cite{EL1} relaxed the local approximation and was an initial motivation for our work~\cite{FHaiLa}.
\end{remark}


\subsection{Connection to BCS theory}

In this subsection we repeat our argument from from \cite{FHaiLa} and describe how the two-body operators \eqref{eq:twobodyk} and $L_{T,W}$ arise in a problem in superconductivity. Our purpose here is to give a motivation and our presentation in this subsection will be informal. For background and references on the mathematical study of BCS theory we refer to our earlier works \cite{HHSS,FHNS,HaSe,FHSS,FHSS2,FrLe,FHaiLa} and, in particular, to the review \cite{HSreview}.

We consider a superconducting sample occupying all of $\R^3$ at inverse temperature $\beta>0$ and chemical potential $\mu\in\R$. The particles interact through a two-body potential $-2V(x-y)$ and are placed in an external electric field with potential $h^2 W(hx)$. In BCS theory the state of a system is described by two operators $\gamma$ and $\alpha$ in $L^2(\R^3)$, representing the one-body density matrix and the Cooper pair wave function, respectively. The operator $\gamma$ is assumed to be Hermitian and the operator $\alpha$ is assumed to satisfy $\alpha^* =\overline{\alpha}$, where for a general operator $A$ we write $\overline A = \mathcal C A \mathcal C$ with $\mathcal C$ denoting complex conjugation. Moreover, it is assumed that
$$
0\leq \begin{pmatrix}
\gamma & \alpha \\ \overline\alpha & 1-\overline\gamma
\end{pmatrix} \leq 1 \,.
$$

In an equilibrium state the operators $\gamma$ and $\alpha$ satisfy the (non-linear) Bogolubov--de Gennes equation
\begin{align*}
& \begin{pmatrix}
\gamma & \alpha \\ \overline\alpha & 1-\overline\gamma
\end{pmatrix} 
= \left( 1+ \exp\left( \beta H_{\Delta_{V,\alpha}}\right) \right)^{-1} \,, \\
& \qquad\text{where}\qquad
\Delta_{V,\alpha}(x,y) = -2V(x-y)\alpha(x,y)
\qquad\text{and}\qquad
H_{\Delta} = 
\begin{pmatrix}
\ch_W & \Delta \\ \overline\Delta & -\overline\ch_W
\end{pmatrix} \,.
\end{align*}
Here $\Delta$ is considered as an integral operator in $L^2(\R^3)$ with integral kernel $\Delta(x,y)$. Moreover, $\ch_W $ is the one-particle operator introduced in \eqref{eq:hw}.

Note that one solution of the equation is $\gamma= (1+\exp(\beta\ch_W))^{-1}$ and $\alpha=0$. This is the \emph{normal state}. We are interested in the local stability of this solution and therefore will linearize the equation around it.

It is somewhat more convenient to write the equation in the equivalent form
$$
\begin{pmatrix}
\gamma & \alpha \\ \overline\alpha & 1-\overline\gamma
\end{pmatrix} 
= \frac12 - \frac12 \tanh\left( \frac{\beta}2 H_{\Delta_{V,\alpha}}\right) \,.
$$
Then, in view of the partial fraction expansion (also known as Mittag--Leffler series)
$$
\tanh z = \sum_{n\in\Z} \frac{1}{z -\ii (n+1/2)\pi}
$$
(where we write $\sum_{n\in\Z}$ short for $\lim_{N\to\infty} \sum_{n=-N}^N$ for conditionally convergent sums like this one; convergence becomes manifest by combining the $+n$ and $-n$ terms),
$$
\tanh\left( \frac{\beta}2 H_{\Delta}\right) = - \frac2\beta \sum_{n\in\Z} \frac{1}{\ii\omega_n - H_\Delta}
$$
with the \emph{Matsubara frequencies}
\begin{equation}
\label{eq:matsubara}
\omega_n = \pi(2n+1)T \,, \qquad n\in\Z \,.
\end{equation}
Using this formula we can expand the operator $\tanh(\beta H_\Delta/2)$ in powers of $\Delta$. Since
\begin{align*}
\frac{1}{\ii\omega_n - H_\Delta} & = \frac{1}{\ii\omega_n - H_0} + \frac{1}{\ii\omega_n - H_0} \begin{pmatrix} 0 & \Delta \\ \overline\Delta & 0 \end{pmatrix} \frac{1}{\ii\omega_n - H_0} + \ldots \\
& = \begin{pmatrix} (\ii\omega_n - \ch_W)^{-1} & 0 \\ 0 & (\ii\omega_n + \overline\ch_W)^{-1} \end{pmatrix} \\
& \qquad + \begin{pmatrix} 0 & (\ii\omega_n - \ch_W)^{-1}\Delta (\ii\omega_n + \overline\ch_W)^{-1} \\ (\ii\omega_n + \overline\ch_W)^{-1}\overline \Delta (\ii\omega_n - \ch_W)^{-1} & 0\end{pmatrix} + \ldots \,,
\end{align*}
the Bogolubov--de Gennes equation for the Cooper pair wave function becomes
$$
\alpha = \frac1\beta \sum_{n\in\Z} (\ii\omega_n - \ch_W)^{-1}\Delta_{V,\alpha} (\ii\omega_n + \overline\ch_W)^{-1} + \ldots \,,
$$
where $\ldots$ stands for terms that are higher order in $\Delta_{V,\alpha}$. The key observation now is that
\begin{equation}
\label{eq:ltsum}
\frac1\beta \sum_{n\in\Z} (\ii\omega_n - \ch_W)^{-1}\Delta_{V,\alpha} (\ii\omega_n + \overline\ch_W)^{-1} = L_{T,W} V\alpha \,.
\end{equation}
(Here $V\alpha$ on the right side is considered as a two-particle wave function, defined by $(V\alpha)(x,y)=V(x-y)\alpha(x,y)$.) This identity follows by writing
\begin{equation}
\label{eq:zetaidentity}
- \frac2\beta \sum_{n\in\Z} (\ii\omega_n - E)^{-1} (\ii\omega_n + E')^{-1} = 
- \frac{2}{\beta} \sum_{n\in\Z} \frac{1}{E+E'} \left( \frac{1}{\ii\omega_n - E} - \frac{1}{\ii\omega_n + E'} \right)
\end{equation}
and using the partial fraction expansion of $\tanh$ to recognize the right side as $\Xi_\beta(E,E')$.

Thus, the linearized Bogolubov--de Gennes equation becomes
$$
\alpha = L_{T,W} V\alpha \,.
$$
There are two ways to make the operator appearing in this equation self-adjoint. The first one is to apply the operator $L_{T,W}^{-1}$ to both sides and to subtract $V\alpha$. In this way we obtain the operator \eqref{eq:twobodyk}. The other way is to multiply both sides of the equation by $V^{1/2}$, to subtract $V^{1/2} L_{T,W} V\alpha$ and to call $\Phi=V^{1/2}\alpha$. In this way we arrive at the operator $1-V^{1/2} L_{T,W} V^{1/2}$ which appears in our main result, Theorem \ref{main}.

The upshot of this discussion is that positivity of the operator \eqref{eq:twobodyk} (or, equivalently, of $1-V^{1/2} L_{T,W} V^{1/2})$ corresponds to local stability of the normal state and the existence of negative spectrum of \eqref{eq:twobodyk} corresponds to local instability. If we define two critical local temperatures $\overline{T_c^{\rm loc}(h)}$ as the smallest temperature above which the normal state is always stable and $\underline{T_c^{\rm loc}(h)}$ as the largest temperature below which the normal state is never stable, then our theorems says that  both $\overline{T_c^{\rm loc}(h)}$ and $\underline{T_c^{\rm loc}(h)}$ are equal to $T_c(1 - D_c h^2)  + O(h) $ as $h\to 0$.

\subsection*{Acknowledgements}

We thank Edwin Langmann who initiated and co-authored our previous work \cite{FHaiLa} which forms the basis of the present paper. 
We further thank Robert Seiringer and Jan Philip Solovej for our long lasting collaboration on BCS theory. 
Further,  partial support by the U.S. National Science Foundation through grant DMS-1363432 (R.L.F.)  is acknowledged.



\section{A representation formula for the operator $L_{T,W}$}

In this section we derive a useful representation formula for the operator $L_{T,W}$ as a sum over contributions from the individual Matsubara frequencies $\omega_n$ from \eqref{eq:matsubara}. Moreover, we express the formula in terms of center of mass and relative coordinates,
$$
r = x-y \,,
\qquad
X= (x+y)/2 \,.
$$
We recall that the corresponding momenta are denoted by $p_r=-\ii\nabla_r$ and $p_X = -\ii\nabla_X$.

Our starting point is \eqref{eq:ltsum}, which can be written in the form 
\begin{align}\label{eq:repr}
\left( L_{T,W} \Delta \right)(x,y) 
= -\frac2\beta \sum_{n\in\Z} \left( \frac 1{\ii\omega_n - \ch_W} \Delta \frac 1{\ii\omega_n+ \ch_W}\right) (x,y) \,.
\end{align}
(Here we used the fact that $\ch_W=\overline{\ch_W}$.) This formula means that as an operator on $L^2(\R^3 \times \R^3)$ we have
$$
L_{T,W} =- \frac2\beta \sum_{n\in\Z} \frac 1{\ii\omega_n - \ch_{W,x}} \frac 1{\ii\omega_n+ \ch_{W,y}} \,.
$$

The strategy now will be to expand the operators $1/(\ii\omega_n\mp\ch_W)$ with respect to $W$. Clearly the leading term is
$$
L_{T,0} = - \frac2\beta \sum_{n\in\Z} \frac 1{\ii\omega_n - \ch_{0,x}} \frac 1{\ii\omega_n+ \ch_{0,y}}
$$
and the subleading correction is $h^2$ times
\begin{align*}
N_{T,W}:= - \frac2\beta \sum_{n\in\Z} & \left( - \frac 1{\ii\omega_n - \ch_{0,x}} \frac 1{\ii\omega_n+ \ch_{0,y}} W(hy)\frac 1{\ii\omega_n+ \ch_{0,y}} \right. \\
& \quad \left. + \frac 1{\ii\omega_n+ \ch_{0,x}} W(hx)\frac 1{\ii\omega_n+ \ch_{0,x}} \frac 1{\ii\omega_n - \ch_{0,y}} \right).
\end{align*}
The following lemma justifies this formal expansion.

\begin{lemma}\label{lemmadiffw0}
As an operator on $L^2(\R^3 \times \R^3)$ we have 
$$
\| L_{T,W} - L_{T,0} \| \lesssim \beta^3 h^2
$$
and
$$
\| L_{T,W} - L_{T,0} -h^2 N_{T,W} \| \lesssim \beta^5 h^4
$$
\end{lemma}

\begin{proof}
Using the resolvent identity we write
\begin{align*}
\frac 1{\ii\omega_n - \ch_{W,x}} \frac 1{\ii\omega_n+ \ch_{W,y}}
& = \frac 1{\ii\omega_n - \ch_{0,x}} \frac 1{\ii\omega_n+ \ch_{0,y}} \\ 
& \quad -  \frac 1{\ii\omega_n - \ch_{0,x}} \frac 1{\ii\omega_n+ \ch_{W,y}} h^2 W(hy)\frac 1{\ii\omega_n+ \ch_{0,y}} \\
& \quad + \frac 1{\ii\omega_n+ \ch_{W,x}} h^2 W(hx)\frac 1{\ii\omega_n+ \ch_{0,x}} \frac 1{\ii\omega_n - \ch_{W,y}} \,.
\end{align*}
The first term on the right side, when summed with respect to $n$, corresponds to the operator $L_{T,0}$. In the remaining terms we use $W\in L^\infty(\R^3)$ and bound each resolvent in norm by $|\omega_n|^{-1}$. The resulting bound is summable with respect to $n$. This proves the first bound. For the proof of the second bound we expand the resolvents once more.
\end{proof}

In the remainder of this section we will do two things, namely bring the operator $L_{T,0}$ in a more explicit form and extract the leading term from the operator $N_{T,W}$. While in Lemma \ref{lemmadiffw0} we considered $L_{T,W}$ as an operator on $L^2(\R^3 \times \R^3)$, we will from now on restrict it to the subspace $L^2_{\rm symm}(\R^3 \times \R^3)$.

In order to investigate the operator $L_{T,0}$ we denote by $g^{z}$ the integral kernel of $1/(z- \ch_0)$, that is,
$$
\frac 1{z- \ch_0} (x,x') = g^z(x-x') \,. 
$$
Using center-of-mass and relative coordinates we can rewrite \eqref{eq:repr} as
\begin{align*}
\left(L_{T,0} \Delta\right)(X+\frac r2,X-\frac r2)
& = -\frac2\beta \sum_{n\in\Z} \iint_{\R^3\times\R^3} dY ds \, \Delta(Y+\frac s2,Y-\frac s2) \\
& \qquad\qquad\qquad \times g^{\ii\omega_n}(X-Y+\frac{r-s}{2}) g^{-\ii\omega_n}(X-Y - \frac{r-s}{2}) \\
& = \iint_{\R^3\times\R^3} dZ ds \, k_{T}(Z,r-s) \Delta(X-Z+\frac s2,X-Z-\frac s2)
\end{align*}
with 
$$
k_{T}(Z,\rho) :=  -\frac2\beta \sum_{n\in\Z}  g^{\ii\omega_n}(Z+\frac{\rho}{2}) g^{-\ii\omega_n}(Z - \frac{\rho}{2}) \,.
$$
Next, we use the fact that $\psi(X - Z) = (e^{-\ii Z\cdot p_X} \psi)(X)$ to write
\begin{align}
\label{eq:reprproof}
\left(L_{T,0} \Delta\right)(X+\frac r2,X-\frac r2)
= \iint_{\R^3\times\R^3} dZ ds \, k_{T}(Z,r-s) \left( e^{-\ii Z\cdot p_X} \Delta\right)(X+\frac s2,X-\frac s2) \,.
\end{align}
We claim that in this formula we can replace  $e^{-\ii Z\cdot p_X}$ by $\cos(Z\cdot p_X)$. To do so, we change variables $Z\mapsto-Z$, $r\mapsto-r$ and $s\mapsto-s$ and use $\Delta(x,y)=\Delta(y,x)$ and $k_{T}(-Z,-r+s)=k_{T}(Z,r-s)$  in order to obtain the same formula as in \eqref{eq:reprproof}, but with $e^{-\ii Z\cdot p_X}$ replaced by $e^{+\ii Z\cdot p_X}$. Adding the two formulas we finally find
\begin{equation}\label{defpX}
\left( L_{T,0}\Delta\right)(X+\frac r2,X-\frac r2) = \iint_{\R^3\times\R^3} dZ\,ds\, k_{T}(Z,r-s) \left(\cos(Z\cdot p_X)\Delta\right)(X+\frac s2,X-\frac s2).
\end{equation}

Next, we derive a convenient representation of $k_{T}(Z,\rho)$. Setting $\ell=p+q$ and $k=(p-q)/2$ and recalling \eqref{eq:ltsum} and \eqref{eq:zetaidentity}, we calculate
\begin{align}
\label{eq:semest1}
k_{T}(Z,\rho) & = - \frac2\beta \sum_{n\in\Z} \iint_{\R^3\times\R^3} \frac{dp}{(2\pi)^3}\,\frac{dq}{(2\pi)^3} \frac{e^{\ii p\cdot(Z+\frac\rho2)}}{\ii\omega_n-p^2+\mu} \frac{e^{\ii q\cdot(Z-\frac\rho2)}}{\ii\omega_n+q^2-\mu} \notag \\
& = \iint_{\R^3\times\R^3} \frac{dp}{(2\pi)^3}\,\frac{dq}{(2\pi)^3} L(p,q) e^{\ii p\cdot(Z+\frac\rho2)+\ii q\cdot(Z-\frac\rho2)} \notag \\
& = \iint_{\R^3\times\R^3} \frac{d\ell}{(2\pi)^3}\,\frac{dk}{(2\pi)^3} L(k+\frac\ell2,k-\frac\ell2) e^{\ii\ell\cdot Z +\ii k\cdot\rho}
\end{align}
with
\begin{equation}
\label{eq:semestl}
L(p,q) : = \frac{\tanh\frac{\beta(p^2-\mu)}2+\tanh\frac{\beta(q^2-\mu)}2}{p^2-\mu+q^2-\mu} \,.
\end{equation}

Let us explain the intuition for the following. Since the external field is varying on the scale $1/h$, which is much larger than the typical distance of between the particles, each momentum $p_X$ will pick up an additional factor of $h$. Therefore, we expect the leading term in \eqref{defpX} to be given by the corresponding operator with $\cos(Z\cdot p_X)$ replaced by $1$. We will justify this approximation in the following lemma. The next order, namely $-(1/2) (Z\cdot p_X)^2$, which will ultimately give rise to the Laplacian in Ginzburg--Landau theory, will be discussed in the following section.

In order to compute the right side of \eqref{defpX} with $\cos(Z\cdot p_X)$ replaced by $1$, we first compute, using \eqref{eq:semest1},
\begin{align}\label{eq:deltafcn}
\int_{\R^3} dZ \, k_{T}(Z,\rho) = \int_{\R^3} \frac{dk}{(2\pi)^3} L(k,k) e^{\ii k\cdot \rho} \,.
\end{align}
This implies that
$$
\iint_{\R^3\times\R^3} dZ\,ds\, k_{T}(Z,r-s) \Delta(X+\frac s2,X-\frac s2) = \left( \chi_\beta(p_r^2-\mu) \Delta \right)(X+r/2,X-r/2) \,,
$$
that is,
\begin{equation}\label{Lt0exp}
L_{T,0} = \chi_\beta(p_r^2 - \mu) - \int_{\R^3} dZ\, k_{T}(Z) \left(1- \cos(Z\cdot p_X)\right) \,,
\end{equation}
where $k_T(Z)$ denotes the operator in $L^2_{\rm symm}(\R^3)$ with integral kernel $k_T(Z,r-s)$.

We now quantify the replacement of $\cos(Z\cdot p_X)$ by $1$.

\begin{lemma}\label{Ltleading}
$$
\left\| \left( L_{T,0} - \chi_\beta(p_r^2 - \mu) \right)\Delta \right\| \lesssim \beta^3 \left\| p_X^2 \Delta \right\|
$$
\end{lemma}

\begin{proof}
We have to bound the integral on the right side of \eqref{Lt0exp}. For this we consider a single term in the definition of $k_T(Z,\rho)$. For fixed $r\in\R^3$ we estimate using Minkowski's inequality
\begin{align*}
& \left( \int_{\R^3} dX \left| \iint_{\R^3\times\R^3} dZ\,ds\, g^{\ii\omega_n}(Z+(r-s)/2) g^{-\ii\omega_n}(Z-(r-s)/2) \right.\right. \\
& \qquad \qquad \times \left.\left. \phantom{\int_{\R^6}} \!\!\!\!\!\!\!\!\!\! \left((1-\cos(Z\cdot p_X))\Delta\right)(X+s/2,X-s/2) \right|^2 \right)^{1/2} \\
& \leq \iint_{\R^3\times\R^3} dZ\,ds \left| g^{\ii\omega_n}(Z+(r-s)/2) g^{-\ii\omega_n}(Z-(r-s)/2) \right| \\
& \qquad\qquad \times \left( \int_{\R^3} dX \left| \left((1-\cos(Z\cdot p_X))\Delta\right)(X+s/2,X-s/2) \right|^2 \right)^{1/2}.
\end{align*}
Now we bound for fixed $Z,s\in\R^3$
\begin{align*}
& \left( \int_{\R^3} dX \left| \left((1-\cos(Z\cdot p_X))\Delta\right)(X+s/2,X-s/2) \right|^2 \right)^{1/2} \\
& \leq \left\| \frac{1-\cos(Z\cdot p_X)}{(Z\cdot p_X)^2} \right\| \left( \int_{\R^3} dX \left| \left((Z\cdot p_X)^2\Delta\right)(X+s/2,X-s/2) \right|^2 \right)^{1/2} \\
& \lesssim |Z|^2 t(s)
\end{align*}
where
$$
t(s) := \left( \int_{\R^3} dX \left| \left(p_X^2\Delta\right)(X+s/2,X-s/2) \right|^2 \right)^{1/2} \,.
$$
Thus, the quantity we are interested in is bounded by a constant times
$$
\iint_{\R^3\times\R^3} dZ\,ds \left| g^{\ii\omega_n}(Z+(r-s)/2) g^{-\ii\omega_n}(Z-(r-s)/2) \right| |Z|^2 t(s) \,.
$$
Using
$$
|Z|^2 \leq \frac12 \left( \left| Z + \frac{r-s}{2} \right|^2 + \left| Z - \frac{r-s}{2}\right|^2 \right)
$$
we can bound the above quantity by
$$
\frac12 \left( \left( \left( |\cdot|^2 g^{\ii\omega_n}\right)\ast g^{-\ii\omega_n} * t \right)(r) + \left( g^{\ii\omega_n} \ast \left( |\cdot|^2 g^{-\ii\omega_n}\right) * t \right)(r) \right).
$$
The $L^2$ norm of this term with respect to $r$ is bounded according to Young's convolution inequality by
$$
\frac12 \left( \left\| |\cdot|^2 g^{\ii\omega_n} \right\|_1 \left\| g^{-\ii\omega_n} \right\|_1 + \left\| g^{\ii\omega_n} \right\|_1 \left\| |\cdot|^2 g^{-\ii\omega_n} \right\|_1 \right) \|t\|_2 \,.
$$
By \cite[Lemma 9]{FHaiLa} this expression is summable with respect to $n$ and therefore the left side in the lemma is bounded by a constant times $\|t\|_2 = \left\| p_X^2 \Delta \right\|$, as claimed.
\end{proof}

This concludes our discussion of the leading term $L_{T,0}$. We now aim at extracting the leading term from the operator $N_{T,W}$ and we concentrate on a term of the form
$$
\iint_{\R^3\times\R^3} dx' dy'\, \left( \frac 1{\ii\omega_n - \ch_0} W(h\cdot) \frac 1{\ii\omega_n - \ch_0} \right) (x,x')  \frac 1{-\ii\omega_n- \ch_0} (y,y') \Delta(x',y') \,.
$$
We introduce again center of mass and relative coordinates  $X=(x+y)/2$, $r=x-y$, $Y=(x'+y')/2$ and $s=x'-y'$. In order to obtain concise expressions we introduce the abbreviation $$ \zeta_X^r = X + r/2, \quad \zeta_Y^{-s} = Y-s/2,$$
where the second term should just show the consistency of the symbol. With these definitions we obtain 
\begin{align}
&\int_{\R^6} dx' dy' \left(\frac 1{\ii\omega_n - \ch_0} W(h\cdot) \frac 1{\ii\omega_n - \ch_0} \right) (x,x')  \frac 1{-\ii\omega_n- \ch_0} (y,y') \Delta(x',y') \nonumber \\
&=\int_{\R^9} dY ds  dz'  \, g^{\ii\omega_n}(\zeta_X^r-z') W(h z') g^{\ii\omega_n} (z' - \zeta_Y^{s} )  g^{-\ii\omega_n}(\zeta_{X-Y}^{s-r} )\Delta (\zeta_Y^{s} ,\zeta_Y^{-s} ) \nonumber\\
& = \int_{\R^9} dY ds  dz \,  g^{\ii\omega_n}(\frac r2-z) W(h X + hz) g^{\ii\omega_n}(z + \zeta_{X-Y}^{-s} ) g^{-\ii\omega_n} (\zeta_{X-Y}^{s-r} )\Delta (\zeta_Y^{s} ,\zeta_Y^{-s} ) \nonumber\\
& = \int_{\R^9} dZ ds  dz \, g^{\ii\omega_n}(\frac r2-z) W(h X + hz) g^{\ii\omega_n}(z + \zeta_Z^{-s} ) g^{-\ii\omega_n}(\zeta_Z^{s-r} )\Delta (\zeta_{X-Z}^{s} ,\zeta_{X-Z}^{-s})\nonumber \\ \label{Werr}
&= \int_{\R^9} dZ ds  dz \, g^{\ii\omega_n}(\frac r2-z) W(h X + hz) g^{\ii\omega_n} (z + \zeta_Z^{-s})  g^{-\ii\omega_n}(\zeta_{Z}^{s-r})  \left( e^{-\ii Z\cdot p_X} \Delta\right)(\zeta_{X}^{s},\zeta_{X}^{-s}) \,, 
\end{align}
where in the last step we used again 
\begin{align*}
\alpha (\zeta_{X-Z}^{s} ,\zeta_{X-Z}^{-s}) & = \alpha(X-Z +s/2, X-Z - s/2) \\
 &= \left( e^{-\ii Z\cdot p_X}\alpha\right)(X + s/2, X-s/2) \\
&=   \left( e^{-\ii Z\cdot p_X} \alpha\right)(\zeta_{X}^{s},\zeta_{X}^{-s}) \,. 
\end{align*}
We claim that to leading order we can replace $W(hX+hz)$ in this integral by $W(hX)$. Therefore we define
\begin{align}
\left(\tilde N_{T,W} \Delta\right)(\zeta_X^r, \zeta_X^{-r})  := W(hX) \iint_{\R^3\times\R^3} dZ ds \, \ell_{T}(Z,r-s) \left(e^{-\ii Z\cdot p_X} \Delta\right)(\zeta_X^s,\zeta_X^{-s})
\end{align}
and 
\begin{align}\label{equ17}
\ell_{T}(Z,\rho) :=  \frac2\beta \sum_{n\in\Z}  \left ( \left({g^{\ii\omega_n}} \ast {g^{\ii\omega_n}}\right) (\zeta_{Z}^{\rho}) g^{-\ii\omega_n}(\zeta_{Z}^{-\rho}) + 
{g^{\ii\omega_n}} (\zeta_{Z}^{\rho}) \left( {g^{-\ii\omega_n}} \ast {g^{-\ii\omega_n}}\right) (\zeta_{Z}^{-\rho}) \right).
\end{align}

\begin{lemma}\label{lemmaequ17}
$$ 
\left\| \left(N_{T,W} - \tilde N_{T,W}\right) \Delta \right\| \lesssim h \left( \left\| \Delta\right\| + \left\| |r| \Delta\right\|\right) \,.
$$
\end{lemma}

\begin{proof}
In \eqref{Werr}, we write
$$
W(hX + hz) = W(hX) + h \int_0^1 z\cdot \nabla W(hX + t hz) \, dt
$$
and then we have to estimate the norm of the error term coming from the $t$-integral. In order to calculate the $L^2(\R^3 \times \R^3)$-norm of the corresponding expression in the $(X,r)$-variables we first fix $r \in \R^3$ and consider the following term, which has a prefactor of $h$ in front,
\begin{multline}
\left( \int_{\R^3} dX \left | \int_{\R^9} dZ ds  dz \, g^{\ii\omega_n}(\frac r2-z) \int_0^1 z\cdot \nabla W(hX + t hz) dt \right.\right. \times \\ \left. \left. \times g^{\ii\omega_n} (z + \zeta_Z^{-s})  g^{-\ii\omega_n}(\zeta_{Z}^{s-r})  \left( e^{-\ii Z\cdot p_X} \Delta\right)(\zeta_{X}^{s},\zeta_{X}^{-s}) \right|^2 \right)^{1/2} \,.
\end{multline} 
Using Minkowski's inequality we can bound this by 
 \begin{align*}
&\int_{\R^9} dZ ds  dz \, |g^{\ii\omega_n}(\frac r2-z) | | g^{\ii\omega_n} (z + \zeta_Z^{-s}) || g^{-\ii\omega_n}(\zeta_{Z}^{s-r})| \\
& \qquad\qquad  \times \left( \int_{\R^3} dX \left| \int_0^1 z\cdot \nabla W(hX + t hz) dt \left( e^{-\ii Z\cdot p_X} \Delta\right)(\zeta_{X}^{s},\zeta_{X}^{-s}) \right|^2 \right)^{1/2} \\
& \qquad \leq  \int_{\R^9} dZ ds  dz \, |g^{\ii\omega_n}(\frac r2-z) | | g^{\ii\omega_n} (z + \zeta_Z^{-s}) || g^{-\ii\omega_n}(\zeta_{Z}^{s-r})| \\
& \qquad\qquad  \times  |z| \| \nabla W\|_\infty  \left( \int_{\R^3} dX \left | \left(e^{-\ii Z\cdot p_X} \Delta\right)(X+s/2,X-s/2) \right|^2 \right)^{1/2} \\
& \qquad = \int_{\R^9} dZ ds  dz \, |g^{\ii\omega_n}(\frac r2-z) | | g^{\ii\omega_n} (z + \zeta_Z^{-s}) || g^{-\ii\omega_n}(\zeta_{Z}^{s-r})| |z| \|\nabla W\|_\infty m(s) \, 
 \end{align*}
where
$$
m(s) : = \left( \int_{\R^3} dX\, | \Delta(X+s/2,X-s/2)|^2 \right)^{1/2}
$$
and where we used the unitarity of $e^{-\ii Z\cdot p_X}$ in the last equality. 

The inequality
$$  
|z| \leq \frac12 |z-r/2| + \frac12 |z+Z-s/2| + \frac12|Z-(r-s)/2| + \frac12 |s|
$$
leads to four terms, which we bound separately. The term with $|z-r/2|$ can be bounded by
\begin{align*}
& \| \nabla W\|_\infty \int_{\R^9} dZ ds  dz \, |g^{\ii\omega_n}( r/2-z) | |r/2 - z|  | g^{\ii\omega_n} (z + Z- s/2) || g^{-\ii\omega_n}(Z + (s-r)/2)|  m(s) \\ 
& \qquad = \| \nabla W\|_\infty \left( \left| \left|\cdot\right| g^{\ii\omega_n} \right| \ast |g^{\ii\omega_n}| \ast |g^{-\ii\omega_n}| \ast m\right) (r) \,.
\end{align*}
According to Young's inequality, the $L^2$ norm of this term is bounded by $\|\nabla W\|_\infty$ times
$$
\||\cdot| g^{\ii\omega_n}\|_1 \|g^{\ii\omega_n}\|_1 \|g^{-\ii\omega_n}\|_1 \|m\|_2 = \||\cdot| g^{\ii\omega_n}\|_1 \|g^{\ii\omega_n}\|_1 \|g^{-\ii\omega_n}\|_1 \|\Delta\|_2 \,.
$$
According to \cite[Lemma 9]{FHaiLa} this expression is summable with respect to $n$ and therefore the contribution of this term to $\left(N_{T,W} - \tilde N_{T,W}\right) \Delta$ is bounded by a constant times $h \|\Delta\|_2$.

The argument for the terms involving $|z+Z-s/2|$ and $|Z-(r-s)/2|$ is similar.

The term with $|s|$ can be bounded by
\begin{align*}
& \| \nabla W\|_\infty \int_{\R^9} dZ ds  dz \, |g^{\ii\omega_n}( r/2-z) | |s|  | g^{\ii\omega_n} (z + Z- s/2) || g^{-\ii\omega_n}(Z + (s-r)/2)|  m(s) \\ 
& \qquad = \| \nabla W\|_\infty \left(|g^{\ii\omega_n}|\ast |g^{\ii\omega_n}| \ast |g^{-\ii\omega_n}| \ast \left( |\cdot| m\right) \right) (r),
\end{align*}
According to Young's inequality, the $L^2$ norm of this term is bounded by $\|\nabla W\|_\infty$ times
$$
\| g^{\ii\omega_n}\|_1 \|g^{\ii\omega_n}\|_1 \|g^{-\ii\omega_n}\|_1 \left\| \left|\cdot\right| m\right\|_2 = \|g^{\ii\omega_n}\|_1 \|g^{\ii\omega_n}\|_1 \|g^{-\ii\omega_n}\|_1 \left\|\left|\cdot\right|\Delta\right\|_2 \,.
$$
Again by \cite[Lemma 9]{FHaiLa} this expression is summable with respect to $n$ and therefore the contribution of this term to $\left(N_{T,W} - \tilde N_{T,W}\right) \Delta$ is bounded by a constant times $h \left\|\left|\cdot\right|\Delta\right\|_2$. This proves the lemma.
\end{proof}


\section{Representation of $L_{T,W}$ on the states $\Delta = \psi(X) \tau(r)$}

We will argue below that we are able to restrict to a specific class of states, which are of the form
$\Delta(X+r/2,X-r/2) = \psi(X) \tau(r)$. Due to the symmetry of $\Delta$, $\tau$ has to be an even function, but 
in fact we will later see that $\tau$ can be assumed as radial, and for the proof of our main theorem 
$\tau$ will be proportional to $V^{1/2}(r) \phi_*(r)$, where $\phi_*(r)$ is the zero eigenstate of 
$1 - V^{1/2} \chi_{\beta_c}(p_r^2 - \mu) V^{1/2}$. 

The following corollary is an immediate consequence of the bounds in the previous section.

\begin{corollary}\label{reprproduct}
If $\Delta(X+r/2,X-r/2) = \psi(X)\tau(r)$ with $\tau$ even, then
\begin{align}
\label{eq:repr2a}
\langle \Delta, L_{T,W} \Delta \rangle & = \langle\psi,\psi\rangle 
\langle \tau,\chi_\beta(p_r^2-\mu)\tau\rangle \nonumber \\ 
& - \int_{\R^3} dZ\, \langle\psi,(1 - \cos(Z\cdot p_X))\psi\rangle  \iint_{\R^3\times\R^3} dr ds \, \overline{\tau(r)}  k_{T}(Z,r-s) \tau(s)\nonumber \\
& + h^2  \int_{\R^3} dZ\, \langle\psi,W(hX)  e^{-\ii Z\cdot p_X} \psi\rangle  \iint_{\R^3\times\R^3} dr ds \, \overline{ \tau(r)} \ell_{T}(Z,r-s)   \tau(s)\nonumber  \\
& +  O(h^3) \|\psi\|^2 \|\tau\| \||\cdot| \tau\| \,.
\end{align}
\end{corollary}

We remark that with slightly more work we could replace the error term $\|\tau\| \||\cdot| \tau\|$ by $\||\cdot|^{1/2} \tau\|^2$.

The second term on the right side of \eqref{eq:repr2a} is given to leading order by the same expression with $1-\cos(Z\cdot p_X)$ replaced by $(Z\cdot p_X)^2/2$. Under the assumption that $\tau$ is a radial function, we therefore obtain $\langle \psi, p_X^2 \psi \rangle$ times a constant depending on $\tau$.

The third term on the right side of \eqref{eq:repr2a} is given to leading order by the same expression with $e^{-\ii Z\cdot p_X}$ replaced by $1$. We therefore obtain $h^2 \langle\psi, W(h\cdot) \psi \rangle$ times a constant depending on $\tau$.

This tells us that the center-of-mass fluctuations are governed by a one-body operator of the form $c_1 p_X^2 + c_2 h^2 W(hX)$, which is unitarily equivalent to the operator 
$$
h^2 \left( c_1 p_X^2 + c_2 W(X)\right).
$$
The precise value of the constants $c_1,c_2$ depends on the specific choice of $\tau$.

As we will show below, the errors made in these two approximations can be controlled by $\|p_X^2 \psi\|^2$ and $h^2\|p_X \psi\| \|\psi\|$. In order to get an intuition why the error terms are indeed of higher order in $h$ we recall the heuristic picture of our chosen scaling. The external field $W$ varies on the scale $1/h$. Therefore we expect the optimal function $\psi$ to match this behavior and vary as well on the macroscopic scale. More precisely, we expect that 
$\psi$ will be of the form $\psi(X) = h^{3/2} \tilde \psi(hX) $ with a function $\tilde\psi$ which is bounded in $H^2$ uniformly for small $h$. Therefore the error bounds $\|p_X^2 \psi\|^2$ and $h^2\|p_X \psi\| \|\psi\|$ are $o(h^2)$.

Next, we formulate this intuitive picture as a precise mathematical statement. 

\begin{theorem} \label{semest}
There is a constant $C$ such that for $\Delta$ of the form
$$
\Delta(X+r/2,X-r/2)=\psi(X)\tau(r)
$$
with $\tau$ radial, one has
\begin{align}
\label{eq:semest}
& \left| \langle \Delta, L_{T,W} \Delta \rangle -  A^{(0)}_T[\tau] \|\psi\|^2 - A^{(1)}_T[\tau] \langle \psi, p_X^2  \psi\rangle - h^2 A_T^{(2)}[\tau] \langle \psi, W(h \cdot) \psi\rangle\right| \notag \\
& \qquad \leq C \left( \|\tau\|^2 \|p_X^2 \psi\|^2 + h^2 \left\| \tau\right\|^2 \| p_X \psi\| \|\psi\| + h^3 \|\psi\|^2\||\cdot|\tau\| \|\tau\| \right)
\end{align}
with
\begin{align*}
A^{(0)}_T[\tau] & = \beta \int_{\R^3} dp\, |\hat \tau(p)|^2 \ g_0(\beta(p^2-\mu)) \,, \\
A^{(1)}_T[\tau] & = - \frac{\beta^2}{4} \int_{\R^3} dp\, |\hat \tau(p)|^2 \left(g_1(\beta(p^2-\mu)) + \frac23 \beta p^2 g_2(\beta(p^2-\mu)) \right)\, ,\\
A^{(2)}_T[\tau] & = \frac{\beta^2}{4} \int_{\R^3} dp\, |\hat \tau(p)|^2\, g_1(\beta(p^2-\mu))
\end{align*}
in terms of the functions $g_0$, $g_1$ and $g_2$ from \eqref{eq:auxiliary}.
\end{theorem}

\begin{proof}
This theorem is essentially a consequence of \eqref{eq:repr2a}. We first notice that
$$
\langle \tau,\chi_\beta(p_r^2-\mu)\tau\rangle = A_T^{(0)}[\tau] \,.
$$
Moreover, using arguments as in the previous subsection one can verify that
$$
\left| \int_{\R^3} dZ\, \langle\psi,(1 - \cos(Z\cdot p_X) - (Z\cdot p_X)^2/2)\psi\rangle  F_\tau(Z) \right| \lesssim \|\tau\|^2 \|p_X^2\psi\|^2
$$
where we have introduced
$$ 
F_\tau(Z) :=   \iint_{\R^3\times\R^3} dr\, ds\,  \overline{\tau(r)}  k_{T}(Z,r-s) \tau(s)\,.
$$
Since $\tau$ is radial, so is $F_\tau$ and therefore
$$
\frac12 \int_{\R^3} dZ\, \langle\psi,(Z\cdot p_X)^2\psi\rangle  F_\tau(Z)
= \frac{1}{6} \int_{\R^3} dZ\, Z^2 F_\tau(Z) \langle\psi,p_X^2\psi\rangle \,.
$$
Now using \eqref{eq:semest1},
$$
\int_{\R^3} dZ\, Z^2 F_\tau(Z) = -\int_{\R^3} dk \, \nabla_{\ell}^2|_{\ell=0} L(k+\frac\ell2,k-\frac\ell2) |\hat \tau(k)|^2 \,,
$$
and a tedious, but straightforward computation yields
$$
\nabla_{\ell}^2|_{\ell=0} L(k+\frac\ell2,k-\frac\ell2) = - \frac{3\beta^2}2 \left(g_1(\beta(k^2-\mu)) + \frac23 \beta k^2 g_2(\beta(k^2-\mu)) \right),
$$
which shows that
$$
- \frac12 \int_{\R^3} dZ\, \langle\psi,(Z\cdot p_X)^2\psi\rangle  F_\tau(Z)
= A^{(1)}_T[\tau] \langle\psi,p_X^2\psi\rangle \,.
$$

Finally, by estimating $1- e^{\ii Z\cdot p_X}$ we obtain
$$
\left| \int_{\R^3} dZ\, \langle\psi,W(hX)  \left( e^{-\ii Z\cdot p_X} - 1\right) \psi\rangle  G_\tau(Z) \right| \lesssim \|\tau\|^2 \| p_X\psi\| \|\psi\|
$$
with
$$
G_\tau(Z) := \iint_{\R^3\times\R^3} dr ds \, \overline{ \tau(r)} \ell_{T}(Z,r-s)   \tau(s) \,.
$$
Rewriting \eqref{equ17} in Fourier space and summing over the Matsubara frequencies gives 
$$
\int_{\R^3} dZ \,\ell_{T}(Z,\rho) = \int_{\R^3} \frac{dk}{(2\pi)^3} \frac{\beta^2}{4} g_1(\beta(k^2-\mu)) e^{ik\cdot\rho}
$$
and therefore
$$
\int_{\R^3} dZ\, G_\tau(Z) = A^{(2)}_T[\tau] \,.
$$
This concludes the proof of the theorem.
\end{proof}


\section{Lower bound on the critical temperature}

We now provide the \emph{Proof of part (1) of Theorem \ref{main}}, which will be a rather straightforward consequence of Theorem \ref{semest}. We will work under Assumptions \ref{ass2} and \ref{ass0}. Assumption \ref{ass1} is not needed in this part of Theorem \ref{main}.

We fix a parameter $T_1$ with $0<T_1<T_c$ and restrict ourselves to temperatures $T\geq T_1$. We consider functions $\Phi$ in $L^2_{\rm symm}(\R^3\times\R^3)$ of the form
$$
\Phi(x,y)=\phi(x-y)h^{3/2} \psi(h (x+y)/2)\,,
$$
where the functions $\phi\in L^2_{\rm symm}(\R^3)$ and $\psi\in L^2(\R^3)$ are still to be determined. At the moment we only require that $\|\psi\|=1$ and $\|p_X^2\psi\|<\infty$. 

We first assume, in addition, that $T_1\geq T_c-Mh^2$ for some constant $M$ independent of $h$. In this case we choose $\phi$ radial and then, applying the expansion from Theorem~\ref{semest} with $\tau(r)=V(r)^{1/2}\phi(r)$, we find that
\begin{align}\label{eq:lowerboundexp}
\langle\Phi, (1- V^{1/2} L_{T,W} V^{1/2})\Phi\rangle = & \|\phi\|^2 - \langle\tau(r)\psi(X), L_{T,W} \tau(r)\psi(X) \rangle \notag \\
\leq & \|\phi\|^2 - A_T^{(0)}[\tau] - h^2 A_T^{(1)}[\tau] \langle\psi,p_X^2\psi\rangle - h^2 A_T^{(2)}[\tau] \langle\psi,W \psi\rangle \notag \\
& + C h^3 \,.
\end{align}
The constant $C$ here depends only on upper bounds on $\|p_X^2\psi\|$, $\|\tau\|$ and $\||\cdot|\tau\|$ (as well as on $M$). The leading order term on the right side is
\begin{equation}
\label{eq:lowerboundmain}
\|\phi\|^2 - A_T^{(0)}[\tau] = \left\langle \phi, \left( 1- V^{1/2}\chi_\beta(p_r^2-\mu) V^{1/2}\right)\phi\right\rangle.
\end{equation}
We choose
$$
\phi :=(2\pi)^{-3/2} \|\chi_{\beta_c}(p^2-\mu) V^{1/2} \phi_*\| \ \phi_* \,,
$$
which makes \eqref{eq:lowerboundmain} equal to zero at $T=T_c$. With this choice of $\phi$ we therefore obtain
\begin{align}
\label{eq:upperproof1}
\langle\Phi, (1- V^{1/2} L_{T,W} V^{1/2})\Phi\rangle 
\leq & A^{(0)}_{T_c}[\tau] - A^{(0)}_T[\tau] -  h^2 \left( A^{(1)}_T[\tau] \langle\psi,p_X^2 \psi\rangle  + A_T^{(2)}[\tau] \langle\psi,W \psi\rangle\right) \nonumber \\&  + C h^3 \,.
\end{align}

In order to proceed, we note the fact that $\tau = V^{1/2}\phi = (2\pi)^{-3/2} V \alpha_*$, and therefore, in terms of the function $t$ from \eqref{eq:t},
\begin{equation}
\label{eq:upperchoicetau}
\hat\tau = (1/2)(2\pi)^{-3/2} t \,.
\end{equation}
It follows from this identity that
$$
\frac{d}{dT}|_{T=T_c} A^{(0)}_{T}[\tau] = -T_c^{-1} \Lambda_2 \,,
$$
and some simple analysis of the function $g_0$ shows that
$$
A^{(0)}_{T_c}[\tau] - A^{(0)}_T[\tau] \leq -\Lambda_2 \frac{T_c-T}{T_c} + C(T_c-T)^2
$$
for all $T_1\leq T\leq T_c$. Using \eqref{eq:upperchoicetau} once again we also find that
$$
A^{(1)}_{T_c}[\tau] = - \Lambda_0 
\qquad\text{and}\qquad
A^{(2)}_{T_c}[\tau] = - \Lambda_1 \,,
$$
which in turn can be used to prove that
$$
A^{(1)}_T[\tau] \geq -\Lambda_0 - C(T_c-T)
\qquad\text{and}\qquad
\left| A^{(2)}_T[\tau] +\Lambda_1 \right| \leq C(T_c-T)
$$
for all $T_1\leq T\leq T_c$.

Inserting these expansions into \eqref{eq:upperproof1} we obtain
\begin{align*}
\langle\Phi, (1- V^{1/2} L_{T,W} V^{1/2})\Phi\rangle 
\leq & - \Lambda_2 \frac{T_c-T}{T_c} + h^2 \left\langle\psi,\left( \Lambda_0 p_X^2 + \Lambda_1 W \right) \psi\right\rangle + C h^3
\end{align*}
for all $T_1\leq T\leq T_c$. Note that here we used the assumption $T\geq T_c - M h^2$, so that the error terms are independent of $T-T_c$.

In order to conclude the proof we assume first, for the sake of simplicity, that $\inf {\rm spec} \left(\Lambda_0 p^2_X + \Lambda_1 W(X)\right)$ is an eigenvalue. In this case we simply choose $\psi$ to be a corresponding normalized eigenfunction. With this choice we obtain, recalling the definition of $D_c$ from \eqref{eq:dc},
\begin{align*}
\langle\Phi, (1- V^{1/2} L_{T,W} V^{1/2})\Phi\rangle 
\leq & -\Lambda_2 \frac{T_c-T}{T_c} + h^2 \Lambda_2 D_c 
+ C h^3 \,.
\end{align*}
The right side is negative if $T< T_c(1- D_c h^2 + (C/\Lambda_2) h^3)$, as claimed.

In case $\inf {\rm spec} \left(\Lambda_0 p^2_X + \Lambda_1 W(X)\right)$ is not an eigenvalue, we choose a sequence of functions $\psi_h$ with $\|\psi_h\|=1$,
$$
\left\langle \psi_h, \left(\Lambda_0 p_X^2 + \Lambda_1 W(X)\right)\psi_h \right\rangle \leq \Lambda_2 \left( D_c + h \right)
\qquad\text{and}\qquad
\| p_X^2 \psi_h \| \leq C
$$
for some $C$ independent of $h$. Such a sequence is obtained by choosing elements in the spectral subspace of $\Lambda_0 p^2_X + \Lambda_1 W(X)$ corresponding to the intervals $\left[\Lambda_2 D_c,\Lambda_2\left(D_c+h\right)\right]$. Since $\Lambda_0 p^2_X + \Lambda_1 W(X)$ has the same operator domain as $p_X^2$ we conclude that
$$
\| p_X^2 \psi_h \|  \lesssim \left\| \left(\Lambda_0 p^2_X + \Lambda_1 W(X) + C' \right) \psi_h\right\| \leq \Lambda_2\left( D_c + h \right) + C' \,,
$$
which proves the last requirement.

We can now repeat the proof with $\psi$ replaced by $\psi_h$. Since all constants were uniform in $\psi$ as long as $\|\psi\|=1$ and $\|p_X^2\psi\|\leq C$, we arrive at the same conclusion as before. This proves the assertion in case $T\geq T_c- Mh^2$ for some fixed $M$ independent of $h$.

Thus, in order to complete the proof of part (1) in the theorem, we show that there is an $M>0$ such that if $T<T_c - M h^2$, then there are $\phi$ and $\psi$ such that the $\Phi$ defined as above satisfies $\langle\Phi, (1-V^{1/2} L_{T,W}V^{1/2})\Phi\rangle<0$.

We proceed similarly as before, but use Corollary \ref{reprproduct} instead of Theorem \ref{semest}. By similar, but simpler estimates as in the proof of Theorem \ref{semest} we obtain
\begin{align}\label{eq:lowerboundexp1}
\langle\Phi, (1- V^{1/2} L_{T,W} V^{1/2})\Phi\rangle \leq \|\phi\|^2 - A_T^{(0)}[\tau] + C h^2 \,.
\end{align}
The constant $C$ here depends only on upper bounds on $\|p_X\psi\|$, $\|\tau\|$ and $\||\cdot|\tau\|$ (as well as on $T_1$). Thus the leading term on the right side is again \eqref{eq:lowerboundmain}.

To bound this term, we denote by $\lambda_T$ the largest eigenvalue of $V^{1/2} \chi_\beta(p_r^2-\mu)V^{1/2}$ in $L^2_{\rm symm}(\R^3)$. By definition of $T_c$ we have $\lambda_{T_c}=1$. Since $\beta\mapsto\chi_\beta(E)$ is monotone for any $E$ with positive derivative, we infer by analytic perturbation theory that there is a $c>0$ such that
$$
\lambda_T \geq \lambda_{T_c} + c(T_c-T) = 1 + c(T_c-T)
\qquad\text{for all}\ 0\leq T\leq T_c \,.
$$
Let $\phi_T$ be a normalized eigenfunction of $V^{1/2} \chi_\beta(p_r^2-\mu)V^{1/2}$ corresponding to $\lambda_T$. With $\phi=\phi_T$ and an arbitrary normalized function $\psi$ with $\|p_X^2\psi\|<\infty$ we obtain,
by inserting \eqref{eq:lowerboundmain} into \eqref{eq:lowerboundexp1} and using the above bound,
\begin{align*}
\langle\Phi, (1- V^{1/2} L_{T,W} V^{1/2})\Phi\rangle
\leq 1- \lambda_{T} + C h^2 \leq - c (T_c -T) + C h^2 \, .
\end{align*}
The right side is negative for $T< T_c - (C/c) h^2$, as claimed. This completes the proof of part (1) of Theorem \ref{main}.
\qed


\section{The approximate form of almost minimizers}

In this and the following section we work under Assumptions \ref{ass2}, \ref{ass0} and \ref{ass1}.

\subsection{The decomposition lemma}

The remainder of this paper is devoted to proving an upper bound on the critical temperature. As a preliminary step we prove in this section a decomposition lemma, which says that, if $|T_c-T|\leq C_1 h^2$ and if $\Phi$ satisfies $\langle\Phi,(1-V^{1/2}L_{T,W}V^{1/2})\Phi\rangle\leq C_2h^2$ for some fixed constants $C_1$ and $C_2$ independent of $h$, then $\Phi$ has, up to a controllable error, the same form as the trial function that we used in the proof of the lower bound on the critical temperature. 

\begin{theorem}\label{decomp}
For given constants $C_1,C_2>0$ there are constants $h_0>0$ and $C>0$ such that the following holds. If $T>0$ satisfies $|T-T_c|\leq C_1h^2$, if $\Phi\in L^2_{\rm symm}(\R^3\times\R^3)$ satisfies $\|\Phi\|=1$ and
$$
\langle\Phi, (1- V^{1/2} L_{T,W} V^{1/2})\Phi\rangle \leq C_2 h^2 \,,
$$
and if $\epsilon$ satisfies $\epsilon\in[h^2,h^2_0]$, then there are $\psi_\leq\in L^2(\R^3)$ and $\sigma\in L^2_{\rm symm}(\R^3\times\R^3)$ such that
$$
\Phi(X+r/2,X-r/2) = \phi_*(r) \psi_{\leq}(X)+\sigma \,,
$$
where
\begin{equation}
\label{eq:decomppsibounds}
\| (p_X^2)^{k/2}\psi_\leq \|^2 \leq C \epsilon^{k-1} h^2 \qquad \text{if}\ k\geq 1 \,,
\end{equation}
\begin{equation}
\label{eq:decompsigmabound}
\|\sigma\|^2 \leq C \epsilon^{-1} h^2
\end{equation}
and
\begin{equation}
\label{eq:decompsilower}
1\geq \|\psi_\leq\|^2 \geq 1- C\epsilon^{-1}h^2 \,.
\end{equation}
Moreover, $\psi_\leq\in\ran\1(p_X^2\leq\epsilon)$ and there is a $\psi_>\in L^2(\R^3)\cap\ran\1(p_X^2>\epsilon)$ such that
$$
\sigma_0(X+r/2,X-r/2):= \frac{ \phi_*(r) \cos(p_X\cdot r/2) }{ \sqrt{ \int_{\R^3} |\phi_*(r')|^2 \cos^2(p_X \cdot r'/2)\, dr' }} \psi_>(X)
$$
satisfies
\begin{equation}
\label{eq:diffsigmasigma0}
\|\sigma - \sigma_0\|^2 \leq C h^2 \,. 
\end{equation}
\end{theorem}

Thus, $\Phi$ is of the form $\psi_\leq(X)\phi_*(r)$ up to a small error. The parameter $\epsilon$ provides a momentum cut-off similarly as in \cite{FHSS,FHSS2} and ensures that we have control on the expectation of $(p_X^2)^2$ in $\psi_\leq$.


\subsection{Upper bound on $L_{T,W}$}

Our goal in this subsection is to obtain an operator lower bound on $1-V^{1/2}L_{T,W}V^{1/2}$. In \cite{FHSS,FHSS2} such a bound was proved by means of a relative entropy inequality \cite[Lemma 3]{FHSS}, which controlled a two-particle operator by the sum of two one-particle operators, and by \cite[Lemma 5]{FHSS} which showed that the energy of the system is dominated by the kinetic energy of the center of mass motion. This was sufficient to recover the corresponding a-priori estimates. 
In \cite{FHaiLa} this operator bound was performed in the presence of a constant magnetic field. Following the spirit of \cite{FHSS,FHSS2} we had to come up with new ideas in order to overcome the problems of non-commutativity of the components of the magnetic momentum operator. In the present much simpler situation we can choose a mixture of the two methods  \cite{FHSS,FHSS2} and~\cite{FHaiLa} 

We define the unitary operator
\begin{equation}
\label{eq:u}
U := e^{-\ii p_X\cdot r/2}
\end{equation}
in $L^2(\R^3\times\R^3)$ where, as usual, $r=x-y$ and $X=(x+y)/2$.

\begin{proposition}\label{lowerboundlt}
There is a constant $C>0$ such that for all $T>0$,
\begin{align*}
V^{1/2} L_{T,W} V^{1/2} & \leq \frac12 \left( U V^{1/2} \chi_\beta(p_r^2-\mu) V^{1/2} U^* + U^* V^{1/2} \chi_\beta(p_r^2-\mu) V^{1/2} U \right) \\
& \qquad + C \beta^3 h^2 .
\end{align*}
\end{proposition}

\begin{proof}
Since for any real numbers $E$ and $E'$ one has
$$
\Xi_\beta(E,E') \leq \frac12 \left( \frac{\tanh\frac{\beta E}{2}}{E} +  \frac{\tanh\frac{\beta E'}{2}}{E'} \right) = \frac12 \left( \chi_\beta(E) + \chi_\beta(E') \right),
$$
we have
$$
L_{T,0} =\Xi_\beta(\ch_{0,x},\ch_{0,y}) \leq \frac12 \left(\chi_\beta(\ch_{0,x}) + \chi_\beta(\ch_{0,y}) \right).
$$
In the variables $r=x-y$, $X=(x+y)/2$ we have $p_x = p_r + p_X/2$ and $p_y = p_r - p_X/2$ and therefore
$$
\ch_{0,x} = (p_r + p_X/2)^2 - \mu = U \left( p_r^2 -\mu\right) U^* \,,
\quad
\ch_{0,y} = (p_r - p_X/2)^2 - \mu = U^* \left( p_r^2 -\mu\right) U \,,
$$
so the previous bound can be written as
$$
L_{T,0} \leq \frac12\left( U \chi_\beta(p_r^2-\mu) U^* + U^* \chi_\beta(p_r^2-\mu) U \right).
$$
On the other hand, by Lemma \ref{lemmadiffw0} we have
$$
L_{T,W} \leq L_{T,0} + C \beta^3 h^2 \,.
$$
Since $V$ commutes with $U$ we obtain the claimed bound.
\end{proof}


\subsection{A priori bound on the critical temperature and an operator inequality}\label{sec:opineq}

As a first consequence of Proposition \ref{lowerboundlt} we obtain a rough a-priori upper bound on the critical temperature.

\begin{corollary}\label{aprioriuppertemp}
There are constants $h_0>0$ and $C>0$ such that for all $0< h \leq h_0$ and $T> T_c+Ch^2$ one has
$$
\langle\Phi, (1- V^{1/2} L_{T,W} V^{1/2})\Phi\rangle >0 \,,
$$
unless $\Phi=0$.
\end{corollary}

\begin{proof}
According to Proposition \ref{lowerboundlt} for all $T\geq T_c$,
\begin{align}\label{eq:lowerbound}
1- V^{1/2} L_{T,W} V^{1/2} & \geq 1 - \frac12 \left( U V^{1/2} \chi_\beta(p_r^2-\mu) V^{1/2} U^* + U^* V^{1/2} \chi_\beta(p_r^2-\mu) V^{1/2} U \right) \notag \\
& \qquad - C h^2 \,.
\end{align}

We next recall that the family of operators $V^{1/2} \chi_\beta(p_r^2-\mu) V^{1/2}$ is non-decreasing with respect to $\beta$ and has an eigenvalue $1$ at $\beta=\beta_c$. Moreover, since the function $\chi_\beta(E)$ is strictly increasing with respect to $\beta$ for every $E\in\R$, we learn from analytic perturbation theory that there are $c>0$ and $T_2>T_c$ such that for all $T_c\leq T\leq T_2$,
$$
V^{1/2} \chi_\beta(p_r^2-\mu) V^{1/2} \leq 1 - c(T-T_c)\,.
$$
Again by monotonicity this implies that for all $T\geq T_c$
$$
V^{1/2} \chi_\beta(p_r^2-\mu) V^{1/2} \leq 1 - c\min\{T-T_c,T_2-T_c\}\,.
$$
Inserting this into the lower bound above we conclude that
$$
1- V^{1/2} L_{T,W} V^{1/2} \geq c\min\{T-T_c,T_2-T_c\} - C h^2 \,.
$$
The right side is positive if $T> T_c + (C/c)h^2$ and $h^2 \leq (c/C)(T_2-T_c)$, which proves the corollary.
\end{proof}

As a consequence of this corollary and the lower bound on the critical temperature, from now on we may and will restrict ourselves to temperatures $T$ such that $|T-T_c|$ is bounded by a constant times $h^2$.

Our next goal is to deduce from Proposition \ref{lowerboundlt} a lower bound on the operator $1- V^{1/2} L_{T,W} V^{1/2}$. We recall that by definition of $\beta_c$ the largest eigenvalue of the operator $V^{1/2} \chi_{\beta_c}(p_r^2-\mu) V^{1/2}$ equals one. Moreover, by Assumption \ref{ass1}, this eigenvalue is simple and $\phi_*$ denotes a corresponding real-valued, normalized eigenfunction. We denote by
$$
P:=|\phi_*\rangle\langle\phi_*|
$$
the corresponding projection and write $P^\bot = 1-P$. Since $V^{1/2} \chi_{\beta_c}(p_r^2-\mu) V^{1/2}$ is a compact operator, there is a $\kappa>0$ such that
\begin{equation}
\label{eq:gap}
V^{1/2} \chi_{\beta_c}(p_r^2-\mu) V^{1/2} \leq 1-\kappa P^\bot \,.
\end{equation}
Finally, we introduce the operator
\begin{equation}
\label{eq:defq}
Q:= \frac12\left( U P U^* + U^* P U \right).
\end{equation}
We can now state our operator inequality for $1- V^{1/2} L_{T,W} V^{1/2}$.

\begin{proposition}\label{opineq}
Given $C_1>0$ and $h_0>0$ with $C_1 h_0^2< T_c$, there is a constant $C>0$ such that for all $|T-T_c|\leq C_1h^2$ and $0< h \leq h_0$ one has
\begin{align}\label{eq:opineq}
1- V^{1/2} L_{T,W} V^{1/2} \geq \kappa\left( 1- Q\right) - C h^2 \,.
\end{align}
\end{proposition}

\begin{proof}
Our starting point is again inequality \eqref{eq:lowerbound}, which is valid for all $|T-T_c|\leq C_1 h_0^2$. Since the derivative of $\chi_\beta(E)$ with respect to $T$ is bounded uniformly in $E$ for $T$ away from $0$, we infer that there is a $C'>0$ such that for all $|T-T_c|\leq C_1 h^2_0$ and all $E\in\R$,
\begin{equation}
\label{eq:chiineq}
\left|\chi_\beta(E) - \chi_{\beta_c}(E)\right| \leq C'|T-T_c| \,.
\end{equation}
This, together with the gap inequality \eqref{eq:gap}, implies that for $|T-T_c|\leq C_1 h^2 \leq C_1 h^2_0$,
\begin{align*}
1- V^{1/2} L_{T,W} V^{1/2} & \geq 
1- \frac12 \left( U V^{1/2} \chi_{\beta_c}(p_r^2-\mu) V^{1/2} U^* + U^* V^{1/2} \chi_{\beta_c}(p_r^2-\mu) V^{1/2} U \right) \\
& \qquad\qquad -C''|T-T_c| - C h^2 \\
& \geq \frac\kappa2 \left( U P^\bot U^* + U^* P^\bot U\right) - (C_1 C'' + C) h^2 \\
& = \kappa\left( 1- Q\right) - (C_1 C'' + C) h^2 \,,
\end{align*}
as claimed.
\end{proof}

Next, we observe that for functions $\Phi\in L^2_{\rm symm}(\R^3\times\R^3)$ one can write 
\begin{align*}
& (Q\Phi)(X+r/2,X-r/2) \\ 
& \quad = \phi_*(r) \cos(p_X\cdot r/2) \int_{\R^3} ds\, \overline{\phi_*(s)} \cos(p_X\cdot s/2) \Phi(X+s/2,X-s/2) \\ 
& \quad =: |A_{p_X}\rangle \langle A_{p_X} | \Phi \rangle
\end{align*}
with
$$
A_p(r):= \phi_*(r) \cos(p\cdot r/2) \,.
$$
(More precisely, the expression $|A_{p_X}\rangle \langle A_{p_X} |$ can be written as a direct integral over the center of mass momenta $p_X$. In the case of magnetic fields \cite{FHaiLa} this did not work because the components of the magnetic momentum did not commute.)

Now we use the fact that in each fiber $Q$ can be estimated from above by its largest eigenvalue, hence we immediately conclude that 
\begin{equation}
1 - Q \geq 1 - \langle A_{p_X}|A_{p_X}\rangle = 1 - R 
\end{equation}
with 
\begin{equation}
\label{eq:r}
R := \int_{\R^3}dr\,  |\phi_*(r)|^2 \cos^2(r\cdot p_X/2)
\end{equation}
acting in $L^2(\R^3)$. Since $\cos(r\cdot p_X/2)^2 \leq 1$ and since $\phi_*$ is normalized, we have $R\leq 1$ and therefore $1-R\geq 0$. We now prove a more precise lower bound.

\begin{lemma}\label{rbound}
There are constants $E_0>0$ and $c>0$ such that
$$
1- R  \geq c \ \frac{p_X^2}{E_0+p_X^2} \,.
$$
\end{lemma}

\begin{proof}
All operators involved are diagonal in Fourier space, so for the proof we can consider $p_X$ to be a vector in $\R^3$. Using the normalization of $\phi_*$ we are thus lead to considering the function
$$
1-R(p_X) = \int_{\R^3} dr \, |\phi_*(r)|^2 \left( 1-\cos^2(p_X\cdot r/2)\right) = \int_{\R^3} dr \, |\phi_*(r)|^2 \sin^2(p_X\cdot r/2)\,.
$$

First, we have
$$
\lim_{p_X\to 0} \frac{1-R(p_X)}{p_X^2} = \frac1{12} \int_{\R^3} dr\, |\phi_*(r)|^2 r^2 =:c \,.
$$
(The right side is finite, as shown in \cite{FHSS}.) Therefore, there is a $\delta>0$ such that $1-R(p_X)\geq (c/2) p_X^2$ for $|p_X|\leq\delta$.

Second, by the Riemann--Lebesgue lemma, we have
$$
\lim_{|p_X|\to\infty} \left(1-R(p_X)\right) = \frac12 \,,
$$
and therefore there is an $M>0$ such that $1-R(p_X)\geq 1/4$ for $|p_X|\geq M$.

Since for any $p_X\neq 0$ the function $r\mapsto \sin^2(p_X\cdot r/2)$ vanishes only on a set of measure zero, we have $1-R(p_X)> 0$ for all $p_X\neq 0$. Since $p_X\mapsto R(p_X)$ is continuous, there is a $c'>0$ such that $1-R(p_X)\geq c'$ for all $\delta\leq |p_X|\leq M$. This proves that
$$
1 - R(p_X) \geq \min\{ (c/2) p_X^2, c',1/4\} \,,
$$
which immediately implies the lemma. 
\end{proof}


\subsection{Proof of the decomposition lemma}

As a consequence of Proposition~\ref{opineq} we now deduce a first decomposition result for almost maximizers $\Phi$ of $1-V^{1/2} L_{T,W} V^{1/2}$. 

Let us now define the projection 
$$
P_Q :=  \frac{| A_{p_X} \rangle \langle A_{p_X} |}{\langle A_{p_X}| A_{p_X}\rangle},
$$
where the last expression is again a direct integral over the momenta $p_X$. 
To see how this operator acts define for a given $\Phi\in L^2_{\rm symm}(\R^3\times\R^3)$,
\begin{equation}\label{eq:defpsi}
\psi(X) := \frac{\langle A_{p_X}|}{\|A_{p_X}\| } \Phi =  \int_{\R^3} ds  \frac{ \phi_*(s) \cos(p_X\cdot s/2) }{ \sqrt{ \int_{\R^3} |\phi_*(s')|^2 \cos^2(p_X \cdot s'/2) ds'}} \Phi(X+s/2,X-s/2) \,.
\end{equation}
Then 
$$P_Q \Phi (X+ r/2,X-r/2) = \frac{ \phi_*(r) \cos(p_X\cdot r/2) }{ \sqrt{ \int_{\R^3} |\phi_*(r')|^2 \cos^2(p_X \cdot r'/2) dr' }} \psi(X),$$
and we define $\xi\in L^2_{\rm symm}(\R^3\times\R^3)$ by 
\begin{equation}
\label{eq:defxi}
\Phi = P_Q \Phi + \xi \,.
\end{equation}

With these definitions we can formulate a first version of the decomposition lemma. 

\begin{lemma}\label{decomp1}
Given $C_1,C_2>0$ there are $h_0>0$, $E_0>0$ and $C>0$ with the following properties. If $|T-T_c|\leq C_1 h^2 \leq C_1 h^2_0$ and if $\Phi\in L^2_{\rm symm}(\R^3\times\R^3)$ with $\|\Phi\|=1$ satisfies
\begin{equation}
\label{eq:almostmin}
\langle\Phi, (1- V^{1/2} L_{T,W} V^{1/2})\Phi\rangle \leq C_2 h^2 \,,
\end{equation}
then, with $\psi$ and $\xi$ defined in \eqref{eq:defpsi} and \eqref{eq:defxi},
$$
\left\langle\psi,\frac{p_X^2}{E_0+p_X^2}\psi\right\rangle + \|\xi\|^2 \leq C h^2 \,.
$$
and
$$
1\geq \|\psi\|^2 \geq 1- Ch^2 \,.
$$
\end{lemma}

\begin{proof}
By Proposition \ref{opineq} and assumption \eqref{eq:almostmin} we obtain
\begin{equation}
\label{eq:apriorisize}
\langle\Phi,(1-Q)\Phi\rangle \leq \kappa^{-1} (C+C_2)h^2 \,.
\end{equation}
By construction, for every fixed value $p_X$ of the Fourier transform with respect to $X$, $P_Q \Phi $ and $\xi$ are orthogonal as functions of $r$. Therefore
$$
\langle \Phi, (1- Q) \Phi\rangle = \langle P_Q \Phi, (1- Q) P_Q \Phi\rangle + \| \xi \|^2 \,.
$$
On the other hand, it is easy to see that 
$$
\langle P_Q \Phi, (1- Q) P_Q \Phi\rangle = \langle \psi,(1-R) \psi \rangle \,.
$$
Therefore the lower bound on $1-R$ from Lemma \ref{rbound} implies the first assertion in the lemma.

In order to prove the second assertion, we note that
$$
\|\psi\|^2 = \|P_Q\Phi\|^2 = \|\Phi\|^2 - \|\xi\|^2 = 1-\|\xi\|^2
$$
and use the bound on $\|\xi\|^2$ from the first assertion.
\end{proof}

\begin{proof}[Proof of Theorem \ref{decomp}]
Let $\psi$ be as in Lemma \ref{decomp1}. For $\epsilon\in[h^2,h^2_0]$ we set
$$
\psi_\leq := \1(p_X^2 \leq \epsilon)\psi \,,
\qquad
\psi_> := \1(p_X^2 > \epsilon)\psi \,.
$$
Recall from Lemma \ref{decomp1} that $\langle\psi,p_X^2(E_0+p_X^2)^{-1}\psi\rangle\leq C h^2$. This implies that for $k\geq 1$,
\begin{align*}
\left\|\left(p_X^2\right)^{k/2} \psi_\leq\right\|^2 & \leq \epsilon^{k-1} \|p_X\psi_\leq\|^2 \leq (E_0+\epsilon) \epsilon^{k-1} \left\langle\psi,\frac{p_X^2}{E_0+p_X^2}\psi\right\rangle \notag \\
& \leq C (E_0+\epsilon) \epsilon^{k-1} h^2
\end{align*}
and
\begin{equation}
\label{eq:apriorioutside}
\|\psi_>\|^2\leq \frac{E_0+\epsilon}{\epsilon}\left\langle\psi,\frac{p_X^2}{E_0+p_X^2}\psi\right\rangle \leq C \frac{E_0+\epsilon}{\epsilon} h^2 \,.
\end{equation}

We now define
$$
\sigma_0(X+r/2,X-r/2):= \frac{ \phi_*(r) \cos(p_X\cdot r/2) }{ \sqrt{ \int_{\R^3} |\phi_*(r')|^2 \cos^2(p_X \cdot r'/2)\, dr' }} \psi_>(X) \,,
$$
$$
\sigma_1(X+r/2,X-r/2):= - \phi_*(r) \left( 1 - \frac{\cos(p_X\cdot r/2) }{ \sqrt{ \int_{\R^3} |\phi_*(r')|^2 \cos^2(p_X \cdot r'/2)\, dr' }} \right) \psi_\leq(X)
$$
and
$$
\sigma:= \sigma_0 + \sigma_1 + \xi \,,
$$
so that, by Lemma \ref{decomp1},
$$
\Phi= \phi_*(r) \psi_{\leq}(X) + \sigma\,.
$$
According to Lemma \ref{decomp1} and \eqref{eq:apriorioutside}, we have
$$
\|\psi_\leq\|^2 = \|\psi\|^2 - \|\psi_>\|^2 \geq 1- Ch^2 - C \epsilon^{-1}h^2\geq 1- C'\epsilon^{-1} h^2 \,.
$$
and, again according to \eqref{eq:apriorioutside}, we have
$$
\|\sigma_0\|^2 = \|\psi_>\|^2 \leq C \frac{E_0+\epsilon}{\epsilon} h^2 \,.
$$
Moreover,
$$
\|\sigma_1\|^2 = \langle \psi_\leq, S \psi_\leq \rangle
$$
with the operator
\begin{align*}
S & := \int_{\R^3} dr\, |\phi_*(r)|^2 \left( 1 - \frac{\cos(p_X\cdot r/2) }{ \sqrt{ \int_{\R^3} |\phi_*(r')|^2 \cos^2(p_X \cdot r'/2)\, dr' }} \right)^2 \\
& = 2 \int_{\R^3} dr\, |\phi_*(r)|^2 \left( 1 - \frac{\cos(p_X\cdot r/2) }{ \sqrt{ \int_{\R^3} |\phi_*(r')|^2 \cos^2(p_X \cdot r'/2)\, dr' }} \right)
\end{align*}
acting in $L^2(\R^3)$. In the same way as in the proof of Lemma \ref{rbound} we can show that
$$
S \leq C \frac{p_X^2}{E_0+p_X^2} \,,
$$
and therefore
$$
\|\sigma_1\|^2 \lesssim \langle \psi_\leq, \frac{p_X^2}{E_0+p_X^2} \psi_\leq \rangle \leq \langle \psi, \frac{p_X^2}{E_0+p_X^2} \psi \rangle \lesssim h^2 \,.
$$
We conclude that
$$
\| \sigma - \sigma_0 \| = \|\sigma_1 + \xi \| \leq \|\sigma_1 \| + \|\xi\| \lesssim h \,.
$$
This concludes the proof of the theorem.
\end{proof}



\section{Upper bound on the critical temperature}

In this section we prove part (2) of Theorem \ref{main}. In view of Corollary \ref{aprioriuppertemp} and the lower bound on the critical temperature it suffices to consider $T$ satisfying $|T-T_c|\leq C_1 h^2$. Moreover, it clearly suffices to consider functions $\Phi$ with $\|\Phi\|=1$ satisfying
$$
\langle\Phi, (1- V^{1/2} L_{T,W} V^{1/2})\Phi\rangle \leq C_2 h^2
$$
(for if there are no such $\Phi$, then the theorem is trivially true). According to Theorem~\ref{decomp}, for any parameter $\epsilon\in [h^2,h^2_0]$, $\Phi$ can be decomposed as
$$
\Phi = \phi_*(r) \psi_\leq(X) + \sigma \,.
$$
Thus,
\begin{align*}
\langle\Phi, (1- V^{1/2} L_{T,W} V^{1/2})\Phi\rangle =   I_1 + I_2 + I_3
\end{align*}
with
\begin{align*}
I_1 &:=  \left\langle \phi_*(r) \psi_\leq(X), \left(1- V^{1/2} L_{T,W} V^{1/2}\right) \phi_*(r) \psi_\leq(X) \right\rangle, \\
I_2 &:= \left\langle \sigma, \left(1- V^{1/2} L_{T,W} V^{1/2}\right) \sigma \right\rangle, \\
I_3 &:= 2\re \left\langle\sigma, \left(1- V^{1/2} L_{T,W} V^{1/2}\right)\phi_*(r) \psi_\leq(X) \right\rangle.
\end{align*}
The term $I_1$ is the main term and can be treated exactly as in the proof of the lower bound on the critical temperature. We obtain
$$
I_1 \geq -\Lambda_2 \frac{T_c-T}{T_c}\|\psi_\leq\|^2 + \left\langle \psi_\leq, \left( \Lambda_0 p_X^2 + \Lambda_1 h^2 W(hX) \right) \psi_\leq \right\rangle
- C \epsilon h^2 \,.
$$
The fact that the error $h^3$ is replaced by $\epsilon h^2$ comes from the bound $\|p_X^2 \psi_\leq\|^2 \lesssim \epsilon h^2$ from \eqref{eq:decomppsibounds}.

Let us therefore bound the error terms $I_2$ and $I_3$. Using the operator inequality from Proposition \ref{opineq}, dropping the non-negative term $\kappa (1-Q)$ and using the bound \eqref{eq:decompsigmabound} on $\sigma$, we obtain
$$
I_2 \gtrsim - h^2 \|\sigma\|^2 \gtrsim - \epsilon^{-1} h^4 \,.
$$

In order to bound $I_3$ we use the first bound in Lemma \ref{lemmadiffw0} and the bounds \eqref{eq:decompsigmabound} and \eqref{eq:decompsilower} on $\sigma$ and $\psi_\leq$ to obtain
\begin{align*}
I_3 & \geq 2\re \left\langle\sigma, \left(1- V^{1/2} L_{T,0} V^{1/2}\right)\phi_*(r) \psi_\leq(X) \right\rangle - C h^2 \|\sigma\| \|\phi_*(r)\psi_\leq(X) \| \\
& \geq 2\re \left\langle\sigma, \left(1- V^{1/2} L_{T,0} V^{1/2}\right)\phi_*(r) \psi_\leq(X) \right\rangle - C' \epsilon^{-1/2} h^3 \,.
\end{align*}
To bound the first term on the right side we decompose $\sigma=\sigma_0 +(\sigma-\sigma_0)$. We claim that
$$
\left\langle\sigma_0, \left(1- V^{1/2} L_{T,0} V^{1/2}\right)\phi_*(r) \psi_\leq(X) \right\rangle =0 \,.
$$
Indeed, to see this, we note that for fixed $r$, the Fourier transforms of $\sigma_0(X+r/2,X-r/2)$ and $V(r)^{1/2} \sigma_0(X+r/2,X-r/2)$ with respect to the variable $X$ are supported in $\{ p_X^2>\epsilon\}$ and likewise the Fourier transforms of $\phi_*(r)\psi_\leq(X)$ and $V(r)^{1/2} \phi_*(r)\psi_\leq(X)$ with respect to the variable $X$ are supported in $\{ p_X^2\leq \epsilon\}$. Thus $\left\langle\sigma_0, \phi_*(r) \psi_\leq(X) \right\rangle =0$, and the full claim follows by observing that the operator $L_{T,0}$ acts diagonally in Fourier space with respect to the $X$ variables, see \eqref{defpX}.

Thus, it remains to bound the term with $\sigma-\sigma_0$. We decompose $L_{T,0}=\chi_{\beta}(p_r^2-\mu) + (L_{T,0} - \chi_{\beta}(p_r^2-\mu))$ and, using the fact that $(1 - V^{1/2} \chi_{\beta_c}(p_r^2-\mu) V^{1/2}) \phi_*=0$, we find
\begin{align*}
& \left\langle \sigma -\sigma_0, \left(1- V^{1/2} L_{T,0} V^{1/2}\right)\phi_*(r) \psi_\leq(X) \right\rangle \\
& \qquad = \left\langle \sigma -\sigma_0, V^{1/2} \left( \chi_{\beta_c}(p_r^2-\mu) - \chi_{\beta}(p_r^2-\mu) \right) V^{1/2} \phi_*(r) \psi_\leq(X) \right\rangle \\
& \qquad \quad - \left\langle \sigma -\sigma_0, V^{1/2} \left(L_{T,0} - \chi_\beta(p_r^2-\mu) \right) V^{1/2} \phi_*(r) \psi_\leq(X) \right\rangle \,.
\end{align*}
Using inequality \eqref{eq:chiineq}, as well as the bounds \eqref{eq:diffsigmasigma0} and \eqref{eq:decompsilower} on $\sigma-\sigma_0$ and $\psi_\leq$, we find
\begin{align*}
\left\langle \sigma -\sigma_0, V^{1/2} \left( \chi_{\beta_c}(p_r^2-\mu) - \chi_{\beta}(p_r^2-\mu) \right) V^{1/2} \phi_*(r) \psi_\leq(X) \right\rangle
& \gtrsim - h^2 \|\sigma-\sigma_0\| \|\psi_\leq \| \\
& \gtrsim - h^3 \,.
\end{align*}
The remaining term we bound similarly using Lemma \ref{Ltleading}, as well as the bounds \eqref{eq:diffsigmasigma0} and \eqref{eq:decompsilower} on $\sigma-\sigma_0$ and $\psi_\leq$,
\begin{align*}
- \left\langle \sigma -\sigma_0, V^{1/2} \left(L_{T,0} - \chi_\beta(p_r^2-\mu) \right) V^{1/2} \phi_*(r) \psi_\leq(X) \right\rangle
& \gtrsim - \|\sigma-\sigma_0\| \|p_X^2\psi_\leq\| \\
& \gtrsim - \epsilon^{1/2} h^2 \,.
\end{align*}

To summarize, we have shown that
\begin{align*}
\langle\Phi, (1- V^{1/2} L_{T,W} V^{1/2})\Phi\rangle 
& \geq 
-\Lambda_2 \frac{T_c-T}{T_c}\|\psi_\leq\|^2 + \left\langle \psi_\leq, \left( \Lambda_0 p_X^2 + \Lambda_1 h^2 W(hX) \right) \psi_\leq \right\rangle \\
& \qquad - C h^2 \left( \epsilon + \epsilon^{-1} h^2 + \epsilon^{-1/2} h + h + \epsilon^{1/2} \right).
\end{align*}
In order to minimize the error we choose $\epsilon = h$. With this choice we obtain, recalling also the lower bound on $\|\psi_\leq\|$ from \eqref{eq:decompsilower},
$$
\langle\Phi, (1- V^{1/2} L_{T,W} V^{1/2})\Phi\rangle 
\geq \left\langle \psi_\leq, \left( \Lambda_0 p_X^2 + \Lambda_1 h^2 W(hX) 
-\Lambda_2 \frac{T_c-T}{T_c} - C' h^{5/2} \right) \psi_\leq \right\rangle
$$
By definition of $D_c$ plus a rescaling we can bound the right side from below by
$$
\left( h^2 \Lambda_2 D_c  - \Lambda_2\frac{T_c-T}{T_c} - C' h^{5/2} \right) \|\psi_\leq\|^2 \,.
$$
Recalling that $\|\psi_\geq\|\neq 0$, we conclude that this is $>0$ provided $T> T_c(1- D_c h^2 + (C'/\Lambda_2) h^{5/2})$. This concludes the proof of the upper bound on the critical temperature.



\bibliographystyle{amsalpha}

\end{document}